\documentclass[11pt]{article}

\usepackage{appendix}
\usepackage{graphicx}
\usepackage{algorithm}
\usepackage{amsmath,amsfonts}
\usepackage{amsmath, amssymb, graphics}
\usepackage{amssymb}
\usepackage{array}
\usepackage{bm}
\usepackage{color}
\usepackage{dsfont}
\usepackage{epsfig}
\usepackage{epstopdf}
\usepackage{etex}
\usepackage{extarrows}
\usepackage{lmodern}
\usepackage{mathrsfs}
\usepackage{pictexwd}
\usepackage{tabularx,ragged2e}
\usepackage{xfrac}

\title{Time value of extra information against its timely value\thanks{I am grateful to Gerhard-Wilhelm Weber and Edward Hoyle for their comments, and to T\"{U}B\.{I}TAK for supporting this research.}}
\author{N. Serhan Aydin} 

\newcommand{\mathsym}[1]{{}}

\newcommand{\F}{\mathcal{F}}

\newcommand{\QM}{\mathbb{Q}}

\newcommand{\E}{\mathbb{E}}
\newcommand{\V}{\mathbb{V}}
\newcommand{\N}{\mathcal{N}}
\newcommand{\s}{\sigma\text{-algebra}}
\newcommand{\Om}{\Omega}
\newcommand{\D}{\Delta}
\newcommand{\nn}{\nonumber}
\newcommand{\ls}{\left[}
\newcommand{\rs}{\right]}
\newcommand{\lr}{\left(}
\newcommand{\rr}{\right)}

\newcommand\independent{\protect\mathpalette{\protect\independenT}{\perp}}
\def\independenT#1#2{\mathrel{\rlap{$#1#2$}\mkern2mu{#1#2}}}
\pdfoutput=1

\newtheorem{theorem}{Theorem}

\newtheorem{corollary}[theorem]{Corollary}
\newtheorem{proposition}[theorem]{Proposition}

\newtheorem{proof}{Proof}

\begin{document}
\maketitle

\begin{abstract}
We introduce an interactive market setup with sequential auctions where agents receive variegated signals with a known deadline. The effects of differential information and mutual learning on the allocation of overall profit \& loss (P\&L) and the pace of price discovery are analysed. We characterise the signal-based expected P\&L of agents based on explicit formulae for the directional quality of the trading signal, and study the optimal trading pattern using dynamic programming and provided that there is a common anticipation by agents of gains from trade. We find evidence in favour of exploiting new information whenever it arrives, and market efficiency. Brief extensions of the problem to risk-adjusted gains as well as risk-averse agents are provided. We then introduce the `information-adjusted risk premium' and recover the signal-based equilibrium price as the weighted average of the signal-based individual prices with respect to the risk-aversion levels.


\end{abstract}



\section{Introduction}
The raison d'\^etre of the markets we study is, in fact, to support information-based trading. Trade can occur on purely informational causes. In \cite{BondEraslan2010}, for example, we are shown that there are situations in which both parties are strictly better off under a trade executed solely on the basis of their individual information. This is somewhat contrary to, e.g., \cite{GrossmanStiglitz1980} and \cite{MilgromStokey1982}. Indeed, one can be overwhelmed by the task of handling a very broad spectrum of aspects where agent-level heterogeneity can arise, such as risk aversion levels, degrees of rationality, patience, beliefs, and information gathering, processing skills, and so on. A detailed classification of different market microstructure models, on the other hand, is given in \cite{Brunnermeier2001} and beyond the scope of this study. However, we start with a review of the selected literature.

Perhaps one of the earliest sequential (discrete) trade models is the one described in the work of Glosten and Milgrom (cf. \cite{GlostenMilgrom1985}), where an attempt is made to explain bid-ask spread as a purely informational phenomenon that is believed to be arising from adverse selection behaviour encountered by less-informed traders. The informational properties of transaction prices and the reaction of the spread to market-generated as well as other public information is also investigated. One of the interesting implications of this model is the possibility of market shutdowns due to severe informational inefficiencies. This is similar to the ``lemons problem'' of Akerlof \cite{Akerlof1970}. The informational content of prices and the value of extra information to the holder are also examined in the work of Kyle (cf. \cite{Kyle1985}) through sequential as well as continuous auction models. Moreover, the latter two seem to converge as the trading interval gets smaller. One interesting result of the model discussed in \cite{Kyle1985}, and to a certain extent in \cite{GlostenMilgrom1985}, is that modelling innovations as functions of quantities traded is found to be consistent with modelling price innovations as the consequence of new information arrivals. The `informativeness of prices' (which is complementary to the amount of information which is yet to be incorporated into prices) in the context of \cite{Kyle1985} refers to the error variance of future dividend given the market clearing price. The question is how intensively the agent, given his superior signal, should trade over time to maximise his profit given his actions might disturb the market (i.e., prices and depth). This model is later on extended in \cite{Back1992} to general continuous distributions for the dividend. Then, a modified version of \cite{GlostenMilgrom1985} in continuous time, where `bluffing' (i.e., mixed strategy) is also allowed, is shown to converge to, again, a modified version of \cite{Kyle1985} with a random signal deadline in \cite{BackBaruch2004}. A rather game-theoretic approach to signal-based trade is taken in \cite{BondEraslan2010} where, this time, the dividend is let endogenously be determined by the action of the agent and its correspondence with the realised fundamental. The signals, in this case, are related to the action that needs to be taken. A sufficient level of signal precision is found to be necessary and sufficient for establishing the case where both seller and buyer are better off from trade in expectation (referred to as ``common knowledge of gains from trade'' in \cite{BondEraslan2010}), which is the equilibrium.

So far, there is no clear mention of the explicit dynamics of information flow, which is the subject of heterogeneity, and it is understood to be an `immediate access' to a publicly unknown value $\phi(X_T)$ without any noise component. Building on \cite{Back1992} and \cite{Kyle1985}, a learning component is added in \cite{BackPedersen1998}. This means the signals are now long-lived with a signal-to-noise varying in time. Although this made possible the mentioning of information `flow' in its true meaning, the interpretation of `learning' through signal in \cite{BackPedersen1998} is slightly different in that when the noise-to-signal, i.e., reciprocal of signal-to-noise, is large, this means the agent is learning a lot. Yet, interestingly, given the total amount of information disparity in favour of the more informed, the pattern in which the information flows is found to be rather irrelevant in equilibrium. Later on, the long-lived signal process is associated with a exponential distributed random deadline (as in \cite{BackBaruch2004} earlier) in \cite{CaldenteyStacchetti2010}. In fact, a random deadline changes the way the strategies for exploiting extra information are structured in various ways, with one way being that agents do not rush to unload their information before it becomes useless and, accordingly, trade frantically as deadline approaches. Backward induction methods of dynamics programming are also rendered inapplicable. 

Perhaps the most interesting alternative to the models of `diverse information' models (where agents do generally share the same probability measure but work over distinct probability spaces) are those of `diverse beliefs'. One way to account for the diversity of beliefs is through equivalent (i.e., defined over the same filtered probability space) probability measures which reflect agents' personal beliefs on the true value of the dividend, as in \cite{BrownRogers2012}. This is maintained by likelihood ratio martingales (or, density processes). Interestingly, the equivalence of the latter two models is established, even without a particular choice of explicit signal structure for private information. And, not so interestingly, the greater the diversity of beliefs, the larger the volume of trade is. A similar approach is found in \cite{CvitanicJouiniMalamudNapp2012} where an equilibrium is established in terms of `surviving agents.' In a belief-heterogeneous market, the surviving agent is found to be the one who is the most rational. Last but not least, in cognisance of the important role played by dynamic optimisation in approaching heterogeneous financial market equilibrium problems, we underline two recent accounts of the latter, i.e. \cite{BussDumas2013} and \cite{DumasLyasoff2012}, where how, in a market of two agents with heterogeneous characteristics, equilibria for various quantities can be found by means of a single backward-induction algorithm is vividly shown.

The approach, in the rest of this paper, to being informationally advantageous or disadvantageous is analogous to the one in \cite{BrodyDavisFriedmanHughston2009}: we do not view the difference between the latter two as having or not having immediate access to the future value of a variable which is unknown to the public information. We rather view it as having access to efficient streams of information or, equivalently (cf. \cite{AdmatiPfleiderer1988}), being more capable to compile and process large and complex datasets out of publicly available information. Both of these are associated with a higher $\sigma$ of $\xi$ in the present context. Yet, in the sense of \cite{BrownRogers2012}, the present framework can also be seen as a diverse-belief model where beliefs are shaped in time by the information itself. 


\section{Modelling Information Flow}
\label{secInfoFlow}
The information-based framework was first introduced in \cite{BrodyHughstonMacrina2007} as a new way of modelling credit risk and, later on, applied to a broad spectrum of issues in financial mathematics, including the valuation of insurance contracts, modelling of defaultable bonds, pricing of inflation-linked assets, and modelling of insider trading, before it was generalised to a wider class of L\'{e}vy information processes in \cite{HoyleHughstonMacrina2011}.

Accordingly, we introduce the signal process $(\xi_{t})_{0\leq t\leq T}$ (or, the information process in the sense of \cite{BrodyHughstonMacrina2007})
\begin{equation}
\label{eqXi}
\xi_{t}= \sigma t X_T + \beta_{t}.
\end{equation}
which is a Brownian random bridge (BRB), as defined in \cite{HoyleHughstonMacrina2011} for the general class of L\'{e}vy processes (i.e., L\'{e}vy random bridges or LRBs). Here, $\beta_{t}$ (or, explicitly $\beta_{0T}^{[0,0]}(t)$) is a standard Brownian bridge over the period $[0,T]$ which takes on value $0$ at the beginning and end, and $\sigma$ is a measure of true signal to noise (henceforth, just `signal-to-noise'). The latter governs the overall speed of revelation of true information about the actual value of the fundamental $X_T$.

We also remark that Eq. \eqref{eqXi} is not the only way to represent information flow. Some other forms have also been considered in the literature with slightly different characteristics, such as, $\xi_{t}=tX_T + \beta_{t}$, $\xi_{t}=(t/T) X_T + \sigma \beta_{t}$ or $\xi_{t}=(t/T) X_T + \beta_{t}$.

More formally, we define a probability space $(\Om,\F,\QM)$, on which the filtration $(\F_t^{\xi})_{t\in {\ls0,T\rs}}$ is constructed. Here, $\QM$, i.e., the risk-neutral measure, is assumed to exist. The default measure is set to $\QM$ throughout the thesis, if not stated otherwise. The filtration $\F_t^\xi$ is generated directly by $(\xi_{s})_{0\leq s\leq t}$ and, thus, simply by $\xi_{t}$ itself. The latter simplification follows from the Markov property of $(\xi_{s})_{0\leq s\leq t}$.

\begin{proposition}
\label{propMarkov}
The process $(\xi_{t})_{0\leq t\leq T}$, as defined in Eq. \eqref{eqXi}, is conditionally Markovian.
\end{proposition}

\begin{proof}
Let $X_T=x$. Defining $B_t$ as a Brownian motion, we can indeed express the signal process $\xi_{t}$ as
\begin{equation}
\label{eqBridgeNew}
\sigma t x + \kappa_t^{-\frac{1}{2}}B_t \quad\text{or}\quad \sigma t x + \kappa_t^{-\frac{1}{2}}\displaystyle\int_0^t \text{d}B_s .
\end{equation}
One can verify that these are identical to
\begin{equation}
\label{eqBridgeNew2}
\xi_{t}= \sigma t x + (T-t)\displaystyle\int_0^t \frac{\text{d}B_s}{T-s},
\end{equation}
which, in turn, implies 
\begin{eqnarray}
\label{eqBridgeNew3}
\text{d}\xi_{t} &=& \lr\sigma x -\displaystyle\int_0^t \frac{\text{d}B_s}{T-s}\rr \text{d}t+ (T-t)\frac{\text{d}B_t}{T-t}\nn\\
&=& \lr\sigma x -\frac{\xi_{t}-\sigma t x}{(T-t)}\rr \text{d}t + \text{d}B_t\nn\\
&=& \lr\sigma x - \xi_{t}/T\rr \kappa_t \text{d}t + \text{d}B_t.
\end{eqnarray}

Equations \eqref{eqBridgeNew2} and \eqref{eqBridgeNew3} indeed follow from two other well-known representations of bridges (see, e.g., \cite{Oksendal1998}). Eq. \eqref{eqBridgeNew3}, on the other hand, directly implies that, given $X_T=x$, $\xi_{t}$ is a Markov process with respect to its own filtration, i.e.,

\begin{equation}
\label{eqMarkov1}
\E [ h(\xi_{t}) | \sigma\lr\xi_{r}\rr_{r\leq s}  ] = \E [ h(\xi_{t}) | \sigma(\xi_{s}) ]\quad (s\leq t),
\end{equation}

for any $x$, and any measurable, finite-valued function $h$ (cf. \cite{Oksendal1998}).
\end{proof}

We are now in a position to work out, with respect to the available information $\F_t^\xi$, the value $S_t$ and dynamics $\text{d}S_t$ of an asset which generates a cashflow $\phi (X_{T})$ at time $T$ for some invertible function $\phi$. The value $S_t$, $0\leq t < T$, is given by
\begin{eqnarray}
\label{eqValue1}
S_t = \delta_t^{-1} \E\ls \phi(X_T) \bigr\rvert \F_t^\xi\rs = \delta_t^{-1} \E\ls \phi(X_T) \bigr\rvert \xi_{t}\rs = \displaystyle\int_\mathbb{X} \phi(x) \pi_t (x) \text{d}x,
\end{eqnarray}
where $\delta_t:=\mathbf{1}_{\left\{t<T\right\}} e^{r (T-t)}$ is the num\'{e}raire and $\pi_t (x):=p(x|\xi_{t})$ the posterior density of the payoff. The quantities $X_T$ and $\phi(X_{T})$ are measurable with respect to $\F_{T}^{\xi}$, but not necessarily w.r.t. $\F_t^{\xi}$, $t<T$. On an important note, we remark that $\beta_{t}$, i.e., the pure noise, is not measurable w.r.t. $\F^{\xi}_t$, meaning that it is not directly accessible to market agents. Thus, an agent, although he observes $\xi_{t}$, cannot separate true signal from noise until time $T$. 

Using Bayesian inference, $S_t$ can be expressed as
\begin{equation}
\label{eqValue4}
S_t = \delta_t^{-1} \frac{\displaystyle\int_\mathbb{X} \phi(x) p(x) e^{\kappa_t\lr\sigma x \xi_{t} - \frac{1}{2} \sigma^2 x^2 t \rr} \text{d}x}{\displaystyle\int_\mathbb{X} p(x) e^{\kappa_t\lr\sigma x \xi_{t} - \frac{1}{2} \sigma^2 x^2 t \rr} \text{d}x},
\end{equation}
with $\kappa_t :=T/(T-t)$ and $\mathbb{X}$ is the support of payoff distribution, whereas the PDE it satisfies as
\begin{eqnarray}
\label{eqDyn7}
\text{d}S_t &=& r S_t \text{d}t + \Lambda_{t} \text{d}W_t,
\end{eqnarray}
with 
$$\Lambda_{t}:= \delta_t^{-1} \sigma \kappa_t {\mathbb{C}\text{ov}}_t \lr \phi (X_T), X_T \rr$$
and
\begin{equation}
\label{eqWt}
\text{d}W_t := \kappa_t \lr \xi_{t}/T - \sigma \phi_t (X_T) \rr \text{d}t + \text{d}\xi_t.
\end{equation}
One can indeed show, by referring to L\'{e}vy's characterisation \cite{Levy1948}, that $W_t$ is a Brownian motion adapted to $F_t^\xi$ (cf. \cite{BrodyHughstonMacrina2007}).

\begin{corollary}
\label{corGaussian}
Assume, as a particular case, that $\phi$ is an identity, i.e., $\phi(x)=x$, and $X_T\sim\N(0,1)$ a priori. Then, Eq. \eqref{eqValue4} implies
\begin{eqnarray}
\label{eqValueGauss}
S_t = \delta_t^{-1}\frac{\sigma \kappa_t \xi_{t}}{\sigma^2 \kappa_t t+1},
\end{eqnarray}
where $\mathbb{X}=(-\infty,\infty)$.
\end{corollary}

\section{Model Setup}
\label{secModel}

We assume that there is a pure dealership market comprising risk-neutral agents with heterogeneous informational access. For simplicity, and w.l.o.g., we assume there is a pair of agents, $j=1,2$, with access to the filtrations generated by $\xi_t^1$, $\xi_t^2$, and a single risky asset with payoff $\phi(X_T)$ not being measurable w.r.t. $\F_t^\xi$, $t<T$. We also assume $\sigma_1<\sigma_2$, i.e., agent 2 is informationally more susceptible than agent 1. In our dynamical model, for simplicity of analysis, we suppose that agents trade with each other futures contracts on the single risky asset at sequential auction times $t_i \in [0,T]$ for $i=1,2,\dots,m$, without any intertemporal consumption and exogenous wealth. Both agents simply follow a buy-and-hold strategy. In this setup, execution of trades, besides a potential profit or loss, results in two things. First, they help, e.g., the central-planner, consolidate information sets of agents at time $t_i$ to have a joint information bundle $\sigma(\xi_{t_i}^1,\xi_{t_i}^2)$. Second, the competitive market price will be discovered immediately. Below we will analyse the latter two separately. Limit orders are cleared by a Walrasian matching engine (as in \cite{BussDumas2013}), which can be deemed a central-planner in the context of \cite{BrownRogers2012} or a group of competitive market makers. The central-planner aims solely to maximise the overall expected profit (or, utility) of agents. 

We also note that, for any given $t$ and a priori density $p(x)$, the price is a function of $\xi$ and $\sigma$, i.e., $S_t = S(\xi,\sigma)$.  This means, if $S_t$ is observed, then one needs to know $\sigma$ to be able to back out $\xi_{t}$. Without knowledge $\sigma$, the observer cannot infer how reliable an observed sample of $\xi_t$ is.

Moments before the sequential auction time $t_{i}$, agents, having observed their signals, submit to the central planner the bid and ask prices at which they are willing to trade. One key property of our model is that an agent may not necessarily know his signal is superior (i.e., agnostic) and the agents will be able to infer each other's prices, and also information (unless they are `omitters', as described below), when, and only if, a price match occurs and a clearing price is set. Otherwise, limit orders are kept with the auction engine (i.e., closed limit order trading book). This also rules out `bluffing' (cf. \cite{BackBaruch2004,CaldenteyStacchetti2010}). 

Individual bid and ask prices are based on the signal-implied prices worked out by virtue of Eq. \eqref{eqValue4} and trade occurs whenever
\begin{equation}
\label{eqTradeRule}
\varsigma^- S_{t}^{1} \geq \varsigma^+ S_{t}^{2} \quad\text{or}\quad \varsigma^- S_{t}^{2} \geq \varsigma^+ S_{t}^{1},
\end{equation}
with $\varsigma^-$ and $\varsigma^+$ being the constant bid and ask multipliers, respectively, where $\varsigma^-\leq 1$ and $\varsigma^+\geq 1$. Obviously, if Eq. \eqref{eqTradeRule} holds with equality, i.e., if $\varsigma^- S_{t}^{1} = \varsigma^+ S_{t}^{2}$ or $\varsigma^- S_{t}^{2} = \varsigma^+ S_{t}^{1}$, then the market price $S_t^\ast$ will be discovered directly. In case of an inequality, under risk-neutrality assumption, the market will clear at the mid-price:
\begin{equation}
\label{eqMidPrice}
S_t^\ast=\frac{\varsigma^- S_{t}^{1}+\varsigma^+ S_{t}^{2}}{2}\quad\text{or}\quad\frac{\varsigma^- S_{t}^{2}+\varsigma^+ S_{t}^{1}}{2}. 
\end{equation}
In \cite{BussDumas2013}, the authors vividly explain why the real-world interpretation of the price posted by a Walrasian auctioneering computer is the bid-ask midpoint.

The initial contract holdings of agents, as denoted by $\theta_0^j$, $j\in\{1,2\}$, are set to $0$. Here $\theta_t^j$ denotes the total time-$t$ net contract stock held by agent $j$. We also define $\theta_t := \sum_j \theta_t^j$ as the total net contract `stock' held by the central clearing at time $t$. Accordingly, total net order `flow' at time $t$ should be $\D\theta_t$ which is given by
\begin{equation}
\D\theta_t=\sum_j \D\theta_t^j= \sum_j q_t^j = q_t
\end{equation}
for some trading process $(q_t^j)_{0\leq t\leq T}$, given by
\begin{equation}
\label{eqXiPiecew}
q_t^j := \left\{ \begin{array}{ll}
      q_t^{j+}, & S_t^j>S_t^\ast, \\
      q_t^{j-}, & S_t^j<S_t^\ast, \\
      0, & \text{otherwise},
\end{array} 
\right.
\end{equation} 
with $q_t^{j+}>0$, $q_t^{j-}<0$. Market clearing conditions imply $q_t=\sum_j q_t^j=0$ and, therefore, $\theta_t=0$, $\forall t\in[0,T]$. Now we define the increasing process $(s_t)_{t\in[0,T]}$, i.e., the time of the last trade prior to time $t$, as follows:
\begin{equation}
s_t = \sup\{s: s<t, |q_s^j| > 0\}.
\end{equation}
It is apparent that $s_t$ is $0$ if $t=0$, or if $t>0$ and $q_s^j=0$ $\forall s\in[0,t)$. The ex-post (i.e., at contract expiry) profit/loss of agent $j$ coming from time $t$ transaction can be written as
\begin{eqnarray}
\label{eqOverallProfit}
\Pi_t^{j} &=& \mathbf{1}_{H\cap \{S_t^j>S_t^\ast\}} q_t^{j+} (X-S_t^\ast) + \mathbf{1}_{H\cap \{S_t^j<S_t^\ast\}} q_t^{j-} (X-S_t^\ast) + \\
&& \mathbf{1}_{L\cap \{S_t^j<S_t^\ast\}} q_t^{j-} (X-S_t^\ast) + \mathbf{1}_{L\cap \{S_t^j>S_t^\ast\}} q_t^{j+} (X-S_t^\ast)\nn\\
\text{\scriptsize{(or, simply)}}&=&q_t^{j} (X-S_t^\ast),
\end{eqnarray}
where $S_t^\ast$ is as in Eq. \eqref{eqMidPrice}, and $H$ and $L$ denote high- and low-type markets, respectively (cf. \cite{Kyle1985}). Eq. \eqref{eqOverallProfit} is based on the correspondence of signal and reality. Market clearing conditions again will require $\Pi_t=\sum_j\Pi_t^j=0$ $\forall t\in[0,T]$.

\section{Numerical Analysis}
\label{secNumerical}

We now present some numerical results based on the setup above. Let $|q_t^j|\in\{0,1\}$ and assume, in this first scenario, that both agents are ``omitters'' (or, ``stubborn bigots'' of \cite{BrownRogers2012}) who never change their mind and simply execute trades according to the following recurring procedure: (1) Observe signal $\xi_t^j$. (2) Quote signal-based bid and ask prices. (3) Let the central-planner determine $-$ using the pre-announced and legally binding matching rule \eqref{eqMidPrice} $-$ the trade direction, if any, and the transaction price (which are then revealed to the agents). Note that agents execute trades ``without'' learning from each other $-$ who could, otherwise, update their information set $-$ and continue to rely solely on their own information sources.
\begin{figure}
\centering
\includegraphics[scale=0.50]{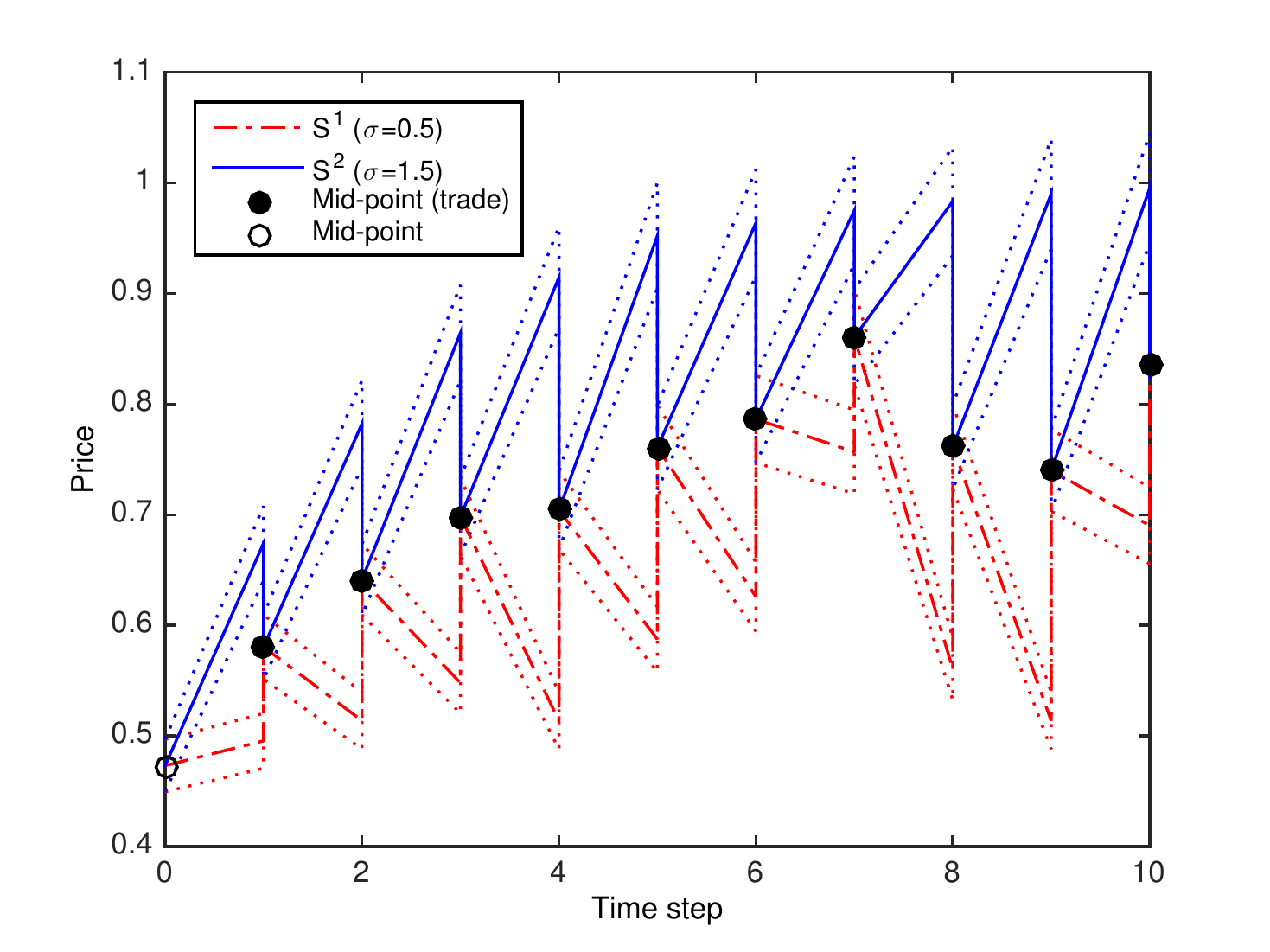}
\caption{Sample evolution of information-based transaction prices in scenario 1 ($|q_t^j|\in\{0,1\}$). Arbitrary parameter values: $T=1$, $\D t = 1/10$, $r=0.05$, $\sigma\in\ls0.5,1.5\rs$, and $\phi(X_T)\in\{0,\bar{1}\}$ (i.e., true value is set to $1$) with $p_0 (x)\in\ls0.5, 0.5\rs$. The dotted lines are bid and ask prices based on $S\varsigma^-$ and $S\varsigma^+$, respectively, with $\varsigma\in\ls 0.95,1.05\rs$.} 
\label{figInteract1}
\end{figure}
\begin{figure}
\centering
\includegraphics[scale=0.50]{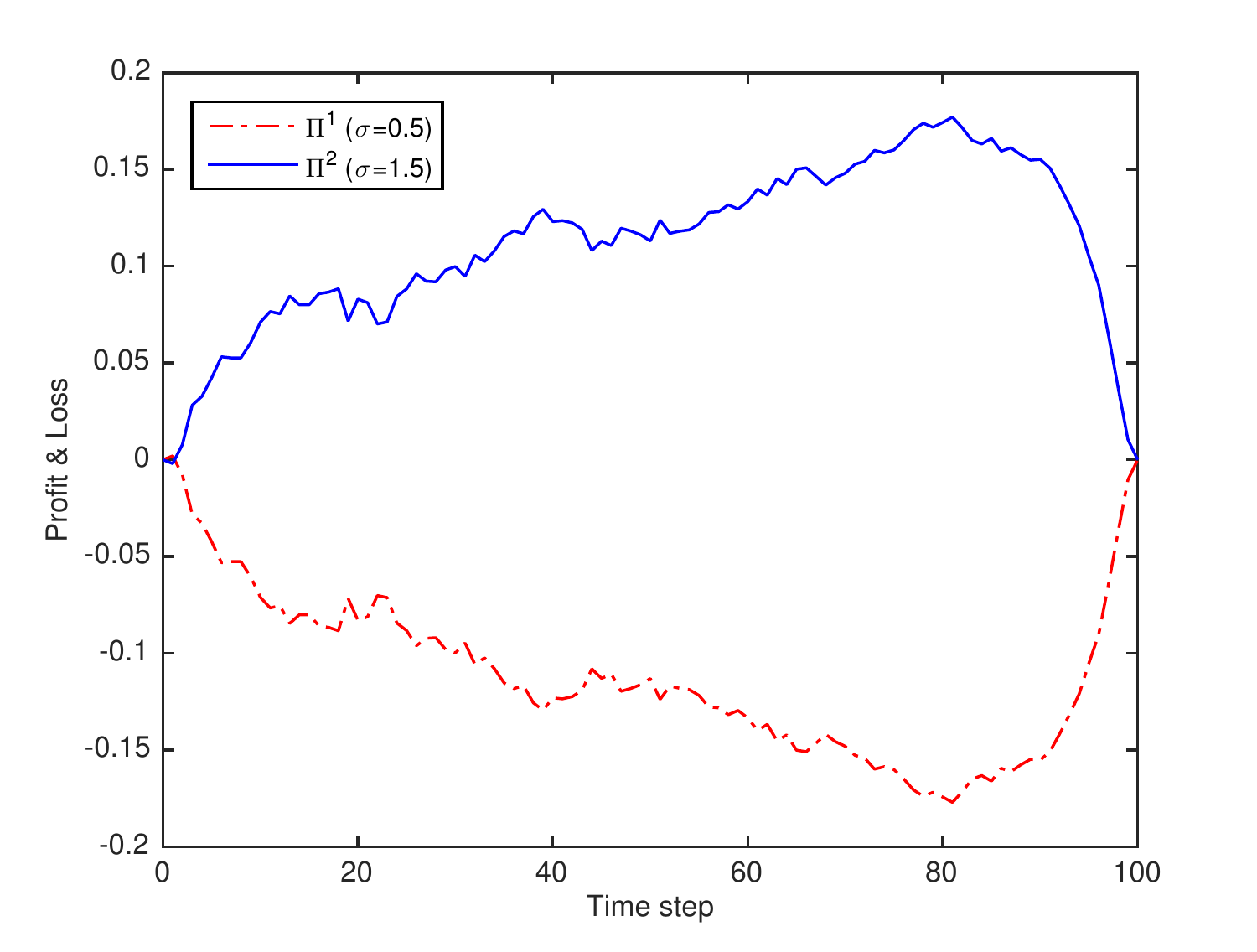}
\caption{Evolution of information-based transaction P\&L averaged over $10^3$ path simulations and based on parameters from Figure \ref{figInteract1}, except that $\D t = 1/100$.} 
\label{figInteract1a}
\end{figure}

In Figure \ref{figInteract1}, where the true fundamental value of $X$ is set to $1$, we illustrate one possible path of such a scenario. Despite a bid-ask margin, occurrence of trade is highly likely in this case as agents do not learn from each other and as personal value judgements diverge. The informationally more (less) susceptible agent, though unknowingly, keeps trading in the right (wrong) direction due to superiority (inferiority) of his signal. Note from Figure \ref{figInteract1} that even after the agent with better signal discovers the asset's true value (around auction $5$), he is still able to execute profitable trades thanks to the matching rule. Figure \ref{figInteract1a}, on the other hand, shows the profit-and-loss (P\&L) results of such a scenario for each time step averaged over $10^3$ simulations, where the number of auctions is increased to $100$. We note at the first glance that the qualitative behaviour of the P\&L agrees with the qualitative behaviour of the magnitude of extra information held (cf. \cite{BrodyDavisFriedmanHughston2009}). 

On an additional note, when multiple ($>$2) agents with various informational capabilities are involved in the market, our numerical results presented in Figure \ref{figInteract1b} suggest that, while P\&L continue to agree with the qualitative behaviour of the magnitude of extra information held by the agent, it is also distributed between agents proportional to the quality of their signal (particularly once the differential informational reaches an adequate level).
\begin{figure}
\centering
\includegraphics[scale=0.50]{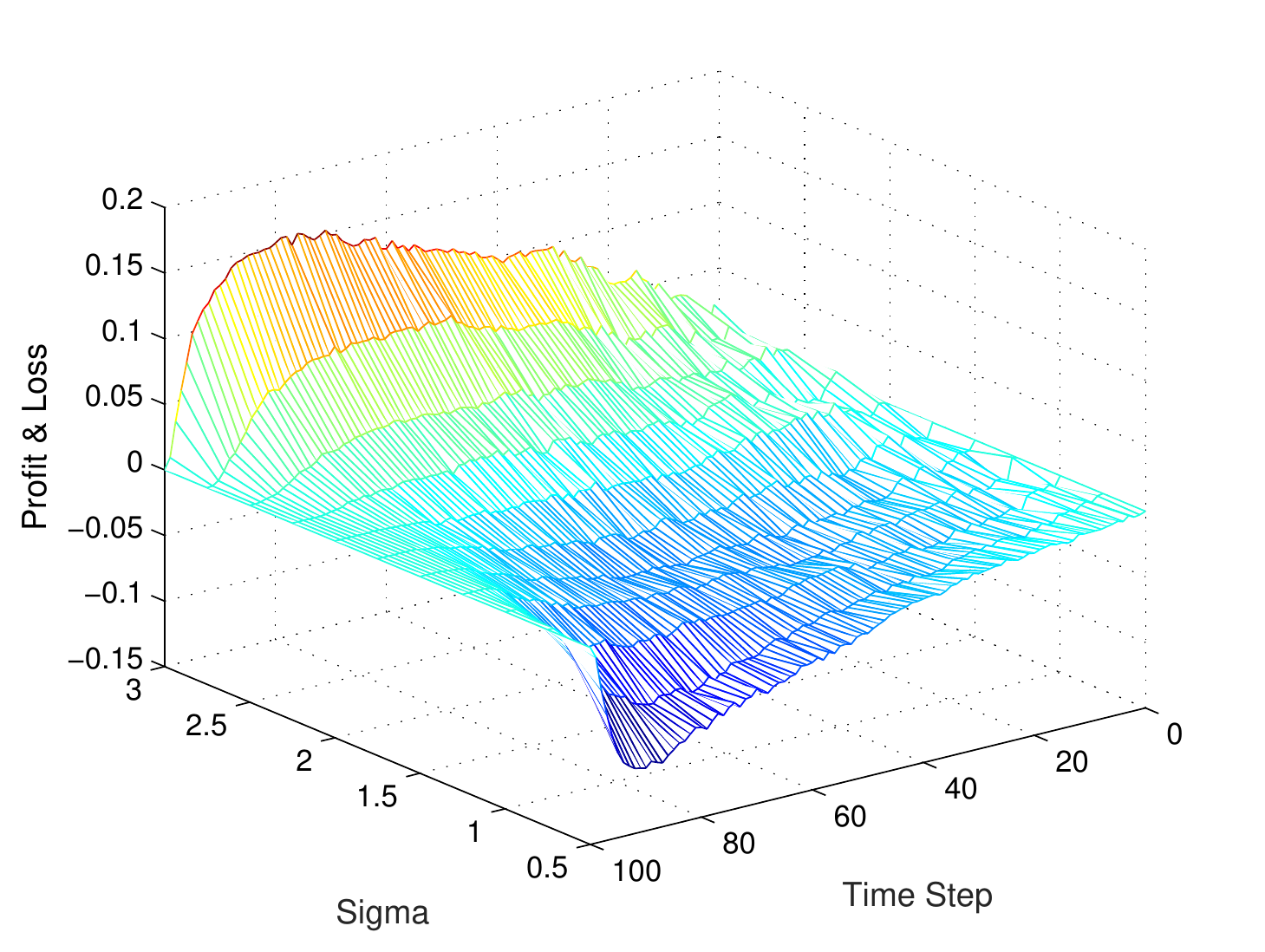}
\caption{Evolution of information-based transaction P\&L of multiple agents averaged over $10^3$ path simulations and based on parameters given in Figure \ref{figInteract1} , except that $\D t = 1/100$.} 
\label{figInteract1b}
\end{figure}

Yet, the exchanges generally do not operate quite this way. A more realistic scenario would be that agents are ``attentive'' and infer their counterpart's posterior $\pi_{s_t}^j (x)$, and, therefore, likelihood $p(\xi_{s_t}^j | x)$, from their price quote at time $s_t$. This would mean having, at any time $t$, partial access to a larger $\s$, e.g., for agent $1$, $\bar{\sigma}(\xi_{t}^1)$, generated by the join of $\sigma(\xi_{t}^1)$ and $\sigma(\xi_{s_t}^2)$, i.e.,
\begin{equation}
\label{eqEffectiveInfo}
\bar{\sigma}(\xi_{t}^1)=\sigma(\xi_{t}^1) \vee \sigma(\xi_{s_t}^2).
\end{equation}

Once agent $1$ gains partial access to $\bar{\sigma}(\xi_{t}^1)$, he updates his posterior from $\pi_{t}^1(x)$ to $\bar{\pi}_{t}^1 (x)$ (by updating $p(\xi_{t}^1 | x)$ to $\bar{p}(\xi_{t}^1 | x)$, i.e., the effective likelihood), which will be again of the form, e.g., for agent 1,
\begin{equation}
\label{eqJointP}
\bar{\pi}_{t}^1 (x) = \frac{p(x) \bar{p}(\xi_{t}^1 | x)}{\displaystyle\int_\mathbb{X} p(x) \bar{p}(\xi_{t}^1 | x) \text{d}x} .
\end{equation}

Note that we intentionally avoid the notations $\sigma(\xi_{t}^1,\xi_{s_t}^2)$ and $p(\xi_{t}^1, \xi_{s_t}^2| x)$ (and use $\bar{\sigma}(\xi_{t}^1)$ and $\bar{p}(\xi_{t}^1 | x)$ instead) so as not to mean that one party's signal is directly observable to the other at the last auction time $s_t$ (which is also not needed).

Thus, before submitting an order at time $t$, having observed a new signal $\xi_t^j$, the agent will need to update his effective information to $\bar{\sigma}(\xi_t^j)=\sigma(\xi_t^1)\vee \sigma(\xi_{s_t}^2)$ (e.g., for agent 1) or $\sigma(\xi_{s_t}^1)\vee \sigma(\xi_t^2)$ (in the case of agent 2). Also, since $\xi_t$ is Markovian, for an agent, partially accessing the signal sample $\xi_{s_t}^j$ of his counterpart will be as valuable as partially accessing his entire signal history $\lr\xi_{s}^{j}\rr_{s\leq s_t}$. Accordingly, right before the auction at time $t$, the `useful' effective likelihood $\bar{p}$ for agent 1 will be
\begin{eqnarray}
\label{eqJointP3}
p(\xi_t^{1},\xi_{s_t}^{2} | x) &=& \frac{1}{2\pi\sqrt{t/\kappa_t}\sqrt{s_t/\kappa_{s_t}}\sqrt{1-\hat{\rho}^2}} \nn\\
&&\cdot \exp\lr-\frac{1}{2}\frac{({s_t}/\kappa_{s_t})(\xi_t^1-\sigma_1 x t)^2}{(1-\hat{\rho}^2)(t/\kappa_t) ({s_t}/\kappa_{s_t})}\rr \nn\\
&&\cdot \exp\lr-\frac{1}{2}\frac{- 2\hat{\rho}(\xi_t^1-\sigma_1 x t)(\xi_{s_t}^2-\sigma_2 x {s_t})\sqrt{(t/\kappa_t)}\sqrt{({s_t}/\kappa_{s_t})}}{(1-\hat{\rho}^2)(t/\kappa_t)({s_t}/\kappa_{s_t})}\rr \nn\\
&&\cdot \exp\lr-\frac{1}{2}\frac{(t/\kappa_t)(\xi_{s_t}^{2}-\sigma_2 x {s_t})^2}{(1-\hat{\rho}^2)(t/\kappa_t) ({s_t}/\kappa_{s_t})}\rr,
\end{eqnarray}
where we used the relation $\beta_{s_t}^1=\rho \beta_{s_t}^2 + \sqrt{1-\rho^2} \bar{\beta}_{s_t}$, with $\beta_{s_t}^2\independent \bar{\beta}_{s_t}$, to find $\hat{\rho}$, i.e., the correlation between $\xi_{t}^{1}$ and $\xi_{{s_t}}^{2}$ conditional on $x$, as follows:
\begin{eqnarray}
\hat{\rho}&=&\frac{\mathbb{C}\text{ov}(\beta_t^1,\beta_{s_t}^2)}{\sigma_{\beta_t^1}\sigma_{\beta_{s_t}^2}}\nn\\
&=&\frac{\mathbb{C}\text{ov}(\rho \beta_t^2 + \sqrt{1-\rho^2} \bar{\beta}_t,\beta_{s_t}^2)}{\sqrt{t/\kappa_t}\sqrt{{s_t}/\kappa_{s_t}}} \nn\\
&=&\rho\frac{{s_t}/\kappa_t}{\sqrt{t/\kappa_t}\sqrt{{s_t}/\kappa_{s_t}}} = \rho \sqrt{\frac{{s_t}}{t}\frac{\kappa_{s_t}}{\kappa_t}},
\end{eqnarray}
with $\rho$ being the correlation between $\xi_{s_t}^{1}$ and $\xi_{{s_t}}^{2}$. We note that $\hat{\rho}$ is a decreasing function of time, as expected, and also that, when $\hat{\rho}=0$, Eq. \eqref{eqJointP3} simply reduces to
\begin{eqnarray}
\label{eqJointP4}
p(\xi_t^{1},\xi_{{s_t}}^{2} | x) &=& \frac{1}{2\pi\sqrt{t/\kappa_t}\sqrt{s/\kappa_{s_t}}} \nn\\
&& \cdot \exp\lr-\frac{1}{2}\frac{({s_t}/\kappa_s)(\xi_t^1-\sigma_1 x t)^2}{(t/\kappa_t) ({s_t}/\kappa_{s_t})}\rr \nn\\
&& \cdot \exp\lr-\frac{1}{2}\frac{(t/\kappa_t)(\xi_{s_t}^{2}-\sigma_2 x {s_t})^2}{(t/\kappa_t) ({s_t}/\kappa_{s_t})}\rr ,
\end{eqnarray}
which also reduces to $p(\xi_t^{1} | x)$ when $s_t=0$ (no trade). The signal-based price of agent $j$, $S_t^j$, is then given by
\begin{equation}
\label{eqOwnPaths}
S_t^j=\E\ls X | \bar{\sigma}(\xi_{t}^j)\rs.
\end{equation}

Accordingly, the new trading procedure is as follows: (1) Observe signal $\xi_t^j$. (1a) Work out $\bar{\sigma}(\xi_t^j)$. (2) Quote signal-based bid and ask prices based on effective information. (3) Let the central-planner do his work (same as (3) above).
\begin{figure}
\centering
\includegraphics[scale=0.50]{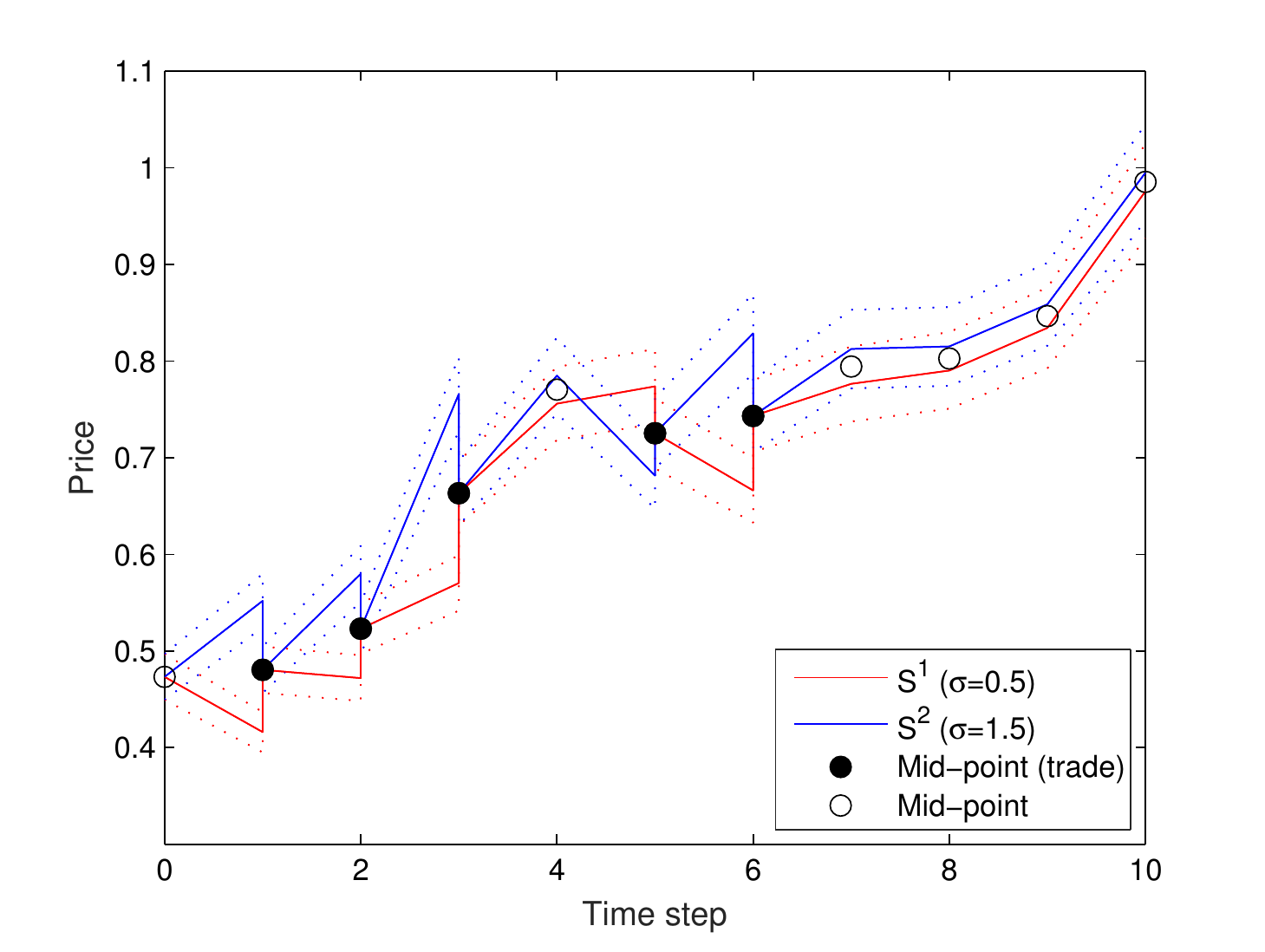}
\caption{Sample evolution of information-based transaction prices along a sample path in scenario 2 ($|q_t^j|\in\{0,1\}$). Arbitrary parameter values: $T=1$, $\D t = 1/10$, $r=0.05$, $\sigma\in\ls0.5,1.5\rs$, and $\phi(X_T)\in\{0,\bar{1}\}$ with $p_0 (x)\in\ls0.5, 0.5\rs$. The dotted lines are bid and ask prices based on $S\varsigma^-$ and $S\varsigma^+$, respectively, with $\varsigma\in\ls 0.95,1.05\rs$.} 
\label{figInteract3}
\end{figure}
\begin{figure}
\centering
\includegraphics[scale=0.50]{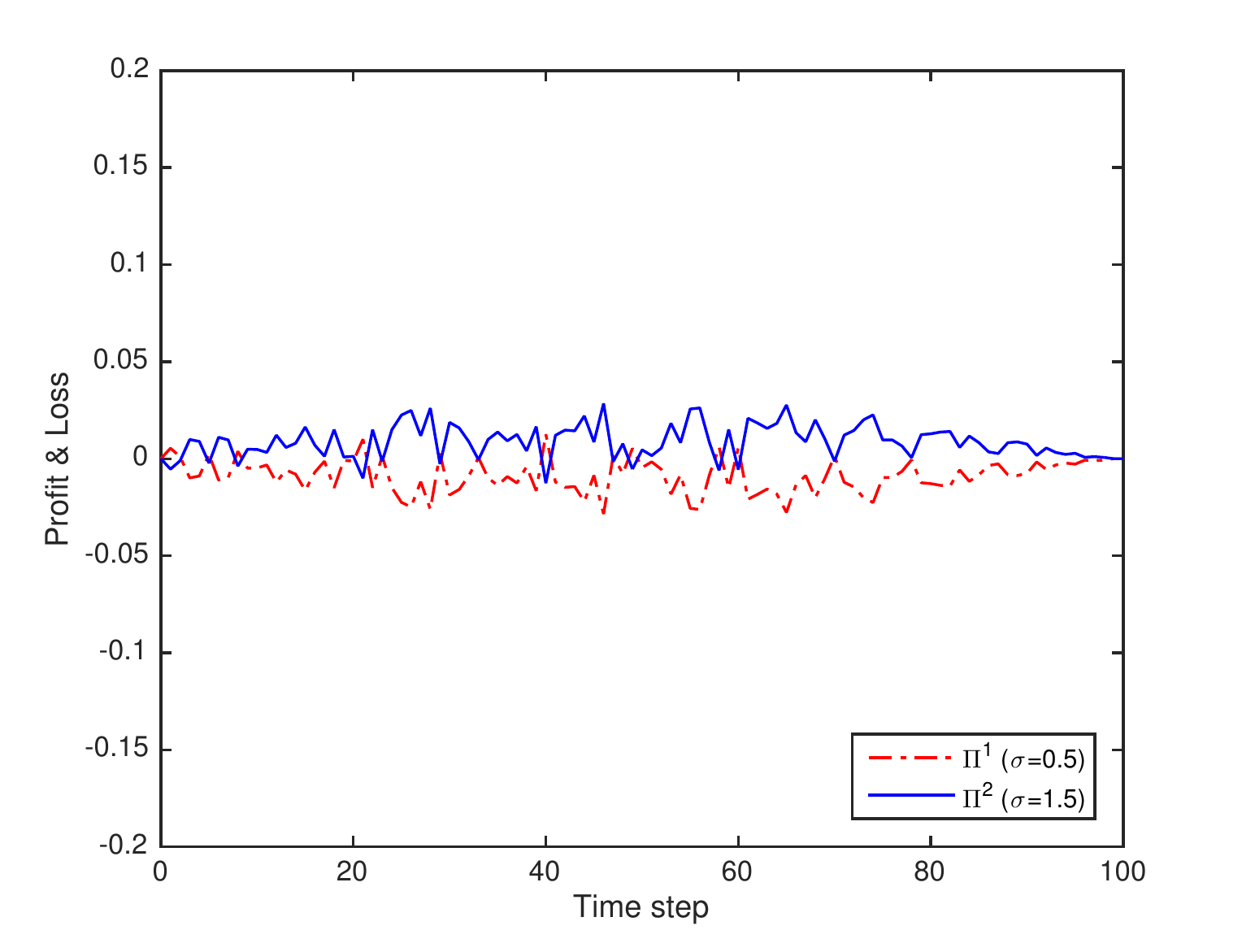}
\caption{Evolution of information-based transaction P\&L averaged over a series of $10^3$ path simulations and based on parameters given in Figure \ref{figInteract3}, except that $\D t = 1/100$.} 
\label{figInteract3a}
\end{figure}

One realisation of this second scenario is depicted in Figure \ref{figInteract3}. At the first glance, learning seems to have decreased profit margins substantially (i.e., to a level where they are often eaten up by the spread, preventing trade). In Figure \ref{figInteract3a}, we again show average stepwise P\&L of agents over $10^3$ realisations. It is apparent from the figure that the informationally more susceptible agent is no more able to extract rents that are as large as in the first scenario (see Figure \ref{figInteract1a}), although he is still able to maintain some modest profits. His ability to maintain modest profits is most likely due to the lag in the learning process as there is still a room for the superior signal to provide the agent receiving it with extra information in-between auctions. The huge difference between the outcomes of two scenarios, i.e., ``omitter'' and ``attentive'', implies that, when each agent deems his own signal superior, there might exist optimal strategies where agents can still be ``attentive'' but, this time, choose which time to reveal their information through trade.
\begin{figure}
\centering
\includegraphics[scale=0.40]{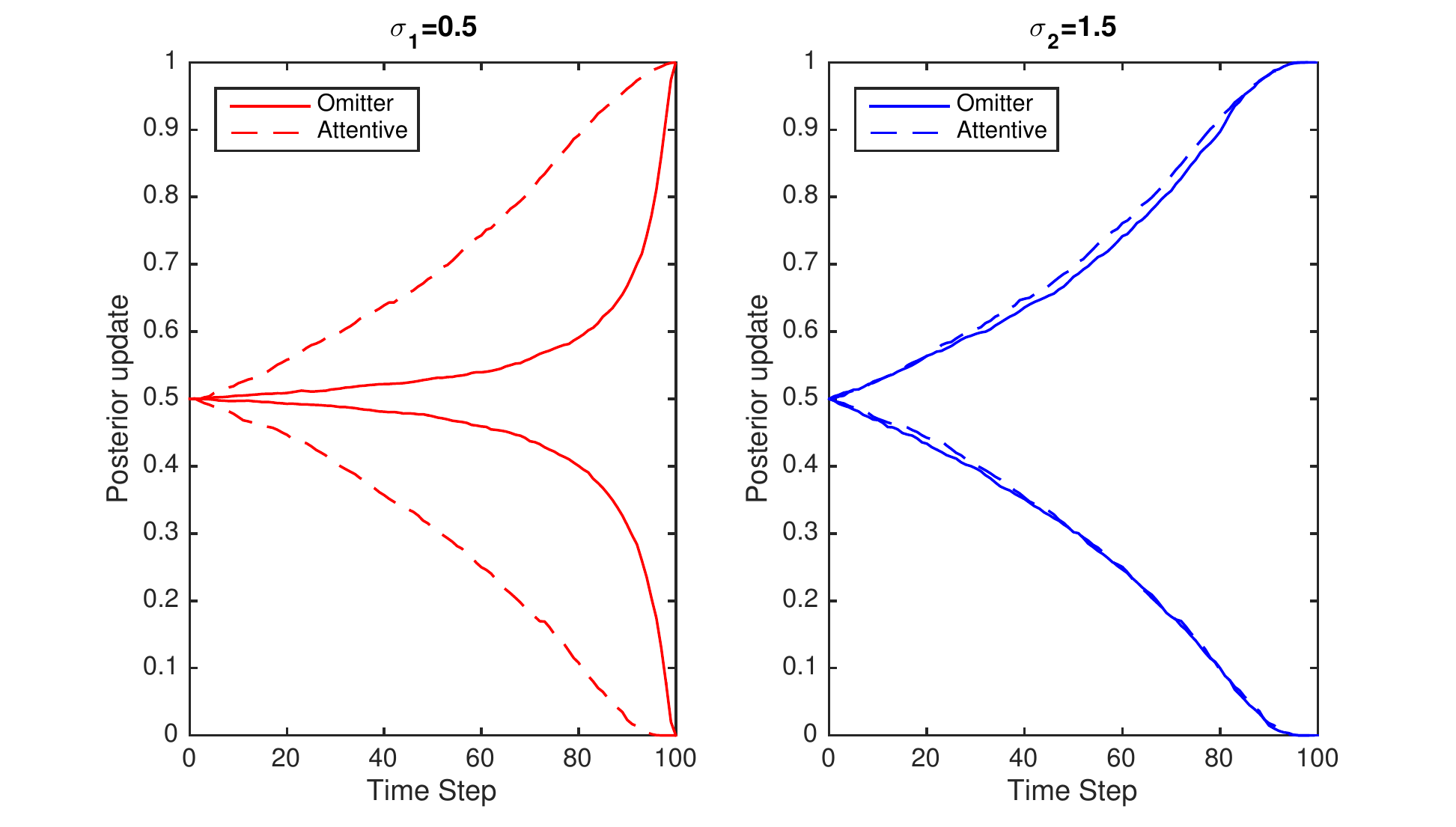}
\caption{Learning process: Bayesian updating of posteriors $\pi_t^j(x)$ averaged over $10^3$ path simulations and based on parameters given in Figure \ref{figInteract3}, except that $\D t = 1/100$.} 
\label{figInteract3b}
\end{figure}
\begin{figure}
\centering
\includegraphics[scale=0.40]{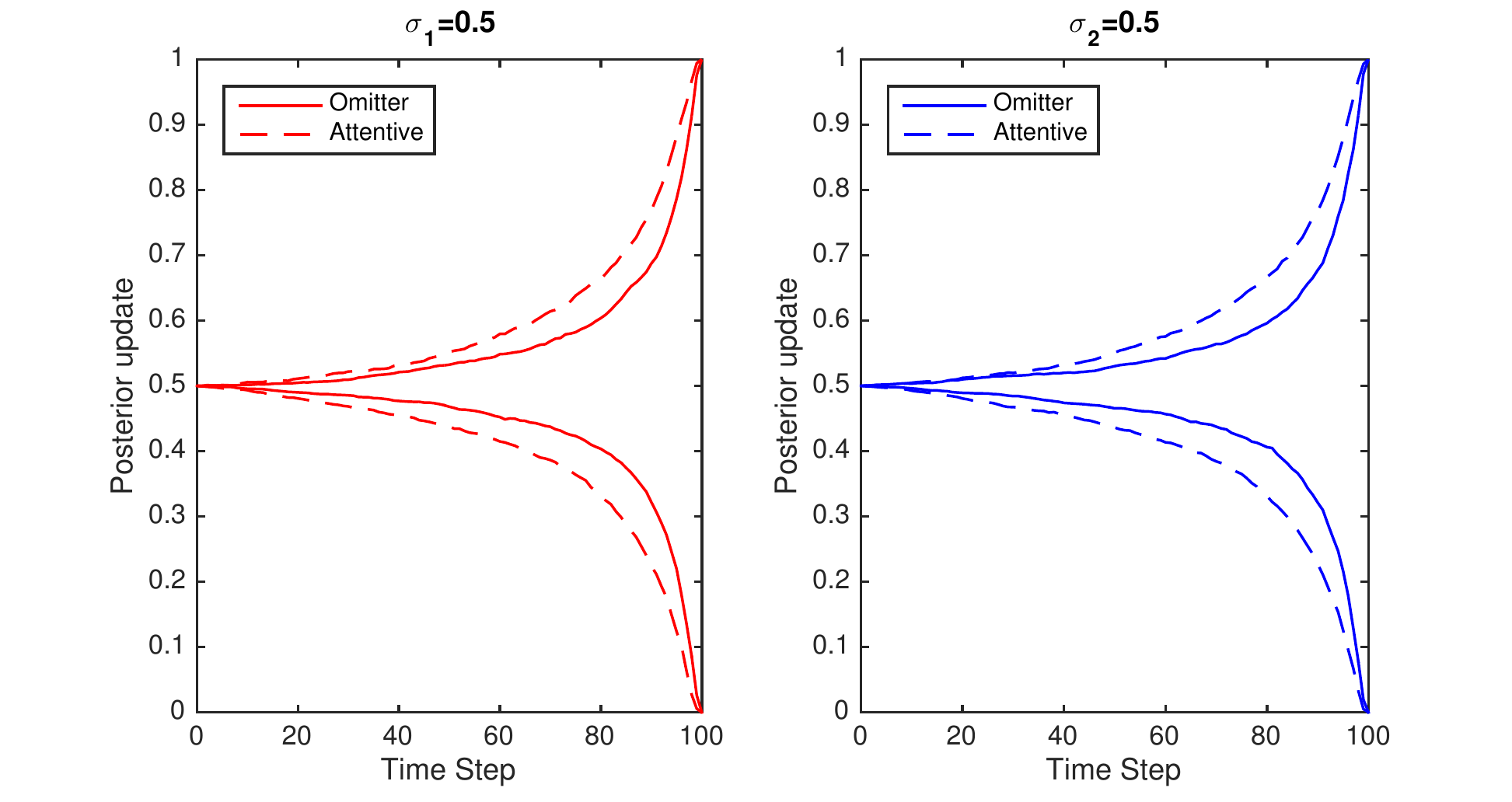}
\caption{Learning process: Bayesian updating of posteriors $\pi_t^j(x)$ averaged over $10^3$ path simulations and based on parameters given in Figure \ref{figInteract3}, except that $\D t = 1/100$.} 
\label{figInteract3c}
\end{figure}

To conclude this section, we compare, in Figures \ref{figInteract3b}-\ref{figInteract3c}, the impact of allowing mutual learning on the speeds at which the two agents discover the true fundamental value of the asset. In the case where the differential between information flow speeds is high (refer to Figure \ref{figInteract3b}), learning seems to work more in favour of the agent with less superior signal with little or no benefit to the agent with a superior signal, whereas, when the differential is minimal (cf. Figure \ref{figInteract3c}), both agents equally benefit from sharing their information via trading.

\section{Signal-based Optimal Strategy}

The P\&L figures provided in Section \ref{secModel} are ex-post, i.e., calculated at the terminal date. In reality, when they trade, agents do so based on their signal-based expectations about the true fundamental value to be revealed at time $T$. They learn whether their earlier trades in futures contracts turned out to be a profit or loss again at time $T$. This, in fact, establishes the main argument which calls for the existence of optimal choices of trading times which maximise their signal-based expected profits: both agents believe that their trades will make them better off (or, there exists `a common knowledge of gains from trade' in the sense of \cite{BondEraslan2010}). Throughout this section, we will regard the agents as `attentive,' and assume $\varsigma^\pm=1$.

\subsection{Characterisation of Expected Profit}

We recall from Section \ref{secNumerical} that, just before the auction at time $t$, the agent $j$ observes the value of his signal and works out his effective information $\bar{\sigma}(\xi_t^j)$ before he makes a judgement of the asset's value. Assuming $X\in\{X^l,X^h\}$ and, again, $|q_t^j|\in\{0,1\}$, the expected (ex-ante) profit of agent $j$ from his possible trade at time $t$ can be decomposed as follows:
\begin{equation}
\label{eqOverallProfitre1}
\E_t^j\ls\Pi_t^j\rs  = P_t^j (\xi_c) \bigr\rvert X^h -\E_t^j \ls S_t^\ast | \xi_c \rs \bigr\rvert - P_t^j (\xi_e) \bigr\rvert X^h -\E_t^j \ls S_t^\ast | \xi_e \rs \bigr\rvert
\end{equation}
with $P_t^j (\xi_c)$ and $P_t^j (\xi_e)$ being the chances of agent $j$ getting correct and erroneous signals at time $t$, respectively. And, again, $\E_t[\cdot]=\E[\cdot|\bar{\sigma}(\xi_t)]$. More formally,
\begin{eqnarray}
\label{eqXicXie1}
P_t^j (\xi_c)&=&P_t^j (H) P_t^j(\xi_c |H) + P_t^j (L) P^j(\xi_c |L) \\
&=&P_t^j (H) P^j(S_t^j>S_t^\ast |H) + P_t^j (L) P^j(S_t^j<S_t^\ast |L), \nn\\
\label{eqXicXie2}
P_t^j (\xi_e)&=&P_t^j (H) P_t^j(\xi_e |H) + P_t^j (L) P^j(\xi_e |L)\\
&=&P_t^j (H) P^j(S_t^j<S_t^\ast |H) + P_t^j (L) P^j(S_t^j>S_t^\ast |L),\nn
\end{eqnarray}
where, again, $H$ and $L$ denote high- and low-type markets in the sense of \cite{Kyle1985}, and
\begin{equation}
\label{eqOverallProfitre1re}
\E_t^j\ls\Pi_t^j\rs  = P_t^j (\xi_c) \bigr\rvert X^h -\E_t^j \ls S_t^\ast | \xi_c \rs \bigr\rvert - P_t^j (\xi_e) \bigr\rvert X^h -\E_t^j \ls S_t^\ast | \xi_e \rs \bigr\rvert
\end{equation}
with $P_t^j (\xi_c)$ and $P_t^j (\xi_e)$ being the chances of agent $j$ getting correct and erroneous signals at time $t$, respectively. And, again, $\E_t[\cdot]=\E[\cdot|\bar{\sigma}(\xi_t)]$. More formally,
\begin{eqnarray}
\label{eqXicXie1a}
P_t^j (\xi_c)&=&P_t^j (H) P_t^j(\xi_c |H) + P_t^j (L) P^j(\xi_c |L) \\
&=&P_t^j (H) P^j(S_t^j>S_t^\ast |H) + P_t^j (L) P^j(S_t^j<S_t^\ast |L), \nn\\
\label{eqXicXie2a}
P_t^j (\xi_e)&=&P_t^j (H) P_t^j(\xi_e |H) + P_t^j (L) P^j(\xi_e |L)\\
&=&P_t^j (H) P^j(S_t^j<S_t^\ast |H) + P_t^j (L) P^j(S_t^j>S_t^\ast |L),\nn
\end{eqnarray}
where, again, $H$ and $L$ denote high- and low-type markets in the sense of \cite{Kyle1985}, and
\begin{eqnarray}
\label{eqXicXie3}
&(\xi_c |H) := (S_t^j>S_t^\ast |H),\quad (\xi_c |L) := (S_t^j<S_t^\ast |L)& \nn\\
&(\xi_e |H) := (S_t^j<S_t^\ast |H),\quad (\xi_e |L) := (S_t^j>S_t^\ast |L).&
\end{eqnarray}
Then, Eq. \eqref{eqOverallProfitre1re} can be written more explicitly as follows:
\begin{eqnarray}
\label{eqOverallProfitre1a}
\E_t^j\ls\Pi_t^j\rs  &=& P_t^j (H) \lr P_t^j(\xi_c |H) \lr X^h -\E_t^j \ls S_t^\ast |H,\xi_c \rs\rr \right.\nn\\
&& \left. - P_t^j(\xi_e |H) \lr X^h -\E_t^j \ls S_t^\ast |H,\xi_e \rs\rr\rr + \nn\\
&&   P_t^j (L) \lr P_t^j(\xi_c |L) \lr \E_t^j \ls S_t^\ast |L,\xi_c \rs -X^l\rr \right. \nn\\
&& \left. -P_t^j(\xi_e |L) \lr \E_t^j \ls S_t^\ast |L,\xi_e \rs -X^l\rr\rr, \nn\\
\end{eqnarray}
where
\begin{eqnarray}
\label{eqHighLowMarket}
P_t^j (H)&=&P_t^j (X^h) = \frac{p_{h} e^{\kappa_t f(t,\sigma_j,\xi^{j},x_{h})}}{\sum_{k\in\{h,l\}} p_{k} e^{\kappa_t f(t,\sigma_j,\xi^{j},x_{k})}},\nn\\
P_t^j (L)&=&P_t^j (X^l) = \frac{p_{l} e^{\kappa_t f(t,\sigma_j,\xi^{j},x_{l})}}{\sum_{k\in\{h,l\}} p_{k} e^{\kappa_t f(t,\sigma_j,\xi^{j},x_{k})}},
\end{eqnarray}
with $f(t,\sigma,\xi,x) :=\sigma\xi_t x - (1/2) \sigma^2 x^2 t$. When the payoff, i.e., $\phi(X)=X$, is continuous, however, Eq. \eqref{eqOverallProfitre1a} implies
\begin{eqnarray}
\label{eqOverallProfitre3}
\E_t^j\ls\Pi_t^j\rs  &=&P_t^j (H) \lr \displaystyle\int_{\mathbb{X}^h} \lr P_t^j(\xi_c |H,x) \lr x-\E_t^j \ls S_t^\ast (x) |H,\xi_c\rs \rr \right. \right. \nn\\
&& \left. \left. - P_t^j(\xi_e|H,x) \lr x- \E_t^j \ls S_t^\ast (x) |H,\xi_e \rs \rr \rr \pi_t^{j+}(x) \text{d}x\rr + \nn\\
&& P_t^j (L) \lr  \displaystyle\int_{\mathbb{X}^l} \lr P_t^j(\xi_c|L,x) \lr \E_t^j \ls S_t^\ast (x) |L,\xi_c)\rs - x \rr \right. \right. \nn\\
&& \left. \left.- P_t^j(\xi_e|L,x) \lr \E_t^j \ls S_t^\ast (x) |L,\xi_e \rs - x \rr \rr \pi_t^{j-}(x) \text{d}x\rr, 
\end{eqnarray}
where $\mathbb{X}^h=(S_0^\ast,X_{\max})$ and $\mathbb{X}^l=(X_{\min}, S_0^\ast)$; $\pi_t^{j+}$ and $\pi_t^{j-}$ are normalised posteriors for high- and low-type markets, respectively; and, at this time,
\begin{equation}
\label{eqHighLowCont}
P_t^j (H) := \displaystyle\int_{\mathbb{X}^h} \pi_t^j(x)\text{d}x,\text{ and }P_t^j (L) := \displaystyle\int_{\mathbb{X}^l} \pi_t^j(x)\text{d}x.
\end{equation}
The notation $S(x)$ is used to denote $\E[X|\xi(x)]$, i.e., the signal-based price of the agent when the actual signal is pinned to the value $x$. In a nutshell, expected profit of the agent is decomposed, through Eq. \eqref{eqOverallProfitre1a} and \eqref{eqOverallProfitre3}, into two components, i.e., whether the agent's signal is pointing at the right (wrong) trade direction and, in that case, what the expected profit (loss) would be.

\subsubsection{Directional Quality of Trading Signal (Digital Dividend)}

Assume, without loss of generality, that $X\in\{x_0,x_1\}$, with $x_0=0$, $x_1>0$ and the prior knowledge of the pair $(p_0,p_1)$. In fact, any binary payoff structure $X\in\{x_0,x_1\}$, $x_1>x_0$ can be simplified as $\{0,x_1-x_0\}$ (a property which will simplify our calculations below). Let the true value of $X$ be $x_1$. Equation \eqref{eqOverallProfitre1a} implies
\begin{eqnarray}
\label{eqEOverallProfit}
\E_t^j\ls\Pi_t^j\rs  &=& P_t^j (x_1)\lr P_t^j(\xi_c|x_1) \lr x_1 -\E_t^j \ls S_t^\ast (x_1) | \xi_c\rs\rr \right.\nn\\
&&\left. - P_t^j(\xi_e|x_1) \lr x_1 -\E_t^j \ls S_t^\ast (x_1) | \xi_e\rs\rr \rr + \nn\\
&& P_t^j (x_0)\lr P_t^j(\xi_c|x_0) \lr \E_t^j \ls S_t^\ast (x_0) | \xi_c \rs - x_0 \rr \right.\nn\\
&&\left. - P_t^j(\xi_e |x_0) \lr \E_t^j \ls S_t^\ast (x_0) | \xi_e \rs - x_0\rr \rr.
\end{eqnarray}
We can calculate the likelihoods of receiving a correct trade signal for agent $1$ when $s_t=s>0$ (i.e., agents did already exchange their information through trading before time $t$) in high- and low-type markets, respectively, as follows:
\begin{eqnarray}
P_t^1(\xi_c |x_1,s) &=& P\lr S_t^1 (x_1,s)>S_t^\ast (x_1,s)\rr = P\lr S_t^1 (x_1,s)/2 > S_t^2 (x_1,s)/2\rr \nn\\
&=& P\lr \frac{\sum_{k=1,2} x_k p_k e^{\kappa_t f(t,\sigma_1,\xi^{1},x_k)+\kappa_s f(s,\sigma_2,\xi^2,x_k)}}{\sum_{k=1,2} p_k e^{\kappa_t f(t,\sigma_1,\xi^{1},x_k)+\kappa_s f(s,\sigma_2,\xi^2,x_k)}} \right.\nn\\
&& > \left. \frac{\sum_{k=1,2} x_k p_k e^{\kappa_t f(t,\sigma_2,\xi^{2},x_k)+\kappa_s f(s,\sigma_1,\xi^1,x_k)}}{\sum_{k=1,2} p_k e^{\kappa_t f(t,\sigma_2,\xi^{2},x_k)+\kappa_s f(s,\sigma_1,\xi^1,x_k)}} \rr .
\end{eqnarray}
A straightforward calculation yields
\begin{eqnarray}
\label{eqP1x1ac}
P_{t}^1(\xi_c |x_1,s) &=& P\lr x_1 p_1 e^{\kappa_t f(t,\sigma_1,\xi^1,x_1)+\kappa_s f(s,\sigma_2,\xi^2,x_1)}/(p_0 +p_1 e^{\kappa_t f(t,\sigma_1,\xi^1,x_1)+\kappa_s f(s,\sigma_2,\xi^2,x_1)})\right.\nn\\
&& \left. > x_1 p_1 e^{\kappa_t f(t,\sigma_2,\xi^2,x_1)+\kappa_s f(s,\sigma_1,\xi^1,x_1)}/( p_0 + p_1 e^{\kappa_t f(t,\sigma_2,\xi^2,x_1)+\kappa_s f(s,\sigma_1,\xi^1,x_1)})\rr \nn\\
&=& P\lr \sigma_1\lr\kappa_t\xi_t^1 -\kappa_s\xi_s^1\rr - \frac{1}{2}\sigma_1^2 x_1 \lr t\kappa_t - s \kappa_s \rr > \sigma_2 \lr \kappa_t\xi_t^2 -\kappa_s\xi_s^2\rr- \frac{1}{2} \sigma_2^2 x_1 \lr t\kappa_t - s\kappa_s \rr \rr \nn\\
&=& P\lr \sigma_1\lr\kappa_t\lr \sigma_1 t x_1 + \beta_t^1\rr -\kappa_s\lr \sigma_1 s x_1 + \beta_s^1\rr\rr  - \sigma_2 \lr \kappa_t \lr \sigma_2 t x_1 + \beta_t^2\rr -\kappa_s \lr \sigma_2 s x_1 + \beta_s^2\rr\rr \right.\nn\\
&& \left. >  \frac{1}{2} x_1 \lr t\kappa_t - s \kappa_s \rr \lr \sigma_1^2 - \sigma_2^2 \rr\rr \nn\\
&=& P\lr \sigma_1\lr\kappa_t \beta_t^1 -\kappa_s \beta_s^1\rr  - \sigma_2 \lr \kappa_t \beta_t^2 -\kappa_s \beta_s^2\rr >  \frac{1}{2} x_1 \lr t\kappa_t - s \kappa_s \rr \lr \sigma_2^2 - \sigma_1^2 \rr\rr \nn\\
&=& P\lr \lr\sigma_1^2\kappa_t^2 t/\kappa_t -\sigma_1^2 \kappa_s^2 s/\kappa_s\rr^\frac{1}{2}z_1  - \lr\sigma_2^2 \kappa_t^2 t/\kappa_t -\sigma_2^2 \kappa_s^2 s/\kappa_s\rr^\frac{1}{2}z_2 \right. \nn\\ 
&& \left. >  \frac{1}{2} x_1 \lr t\kappa_t - s \kappa_s \rr \lr \sigma_2^2 - \sigma_1^2 \rr\rr \nn\\
&=&\Theta\lr\frac{1}{2}\frac{ x_1 (t\kappa_t - s \kappa_s) \lr \sigma_1^2 - \sigma_2^2 \rr}{\lr\sigma_1^2+\sigma_2^2\rr^{\sfrac{1}{2}}(t\kappa_t - s \kappa_s)^{\sfrac{1}{2}}}\rr\quad{(s,t<T)},\nn\\
\end{eqnarray}
where, again, $\Theta(\cdot)$ is the standard normal cumulative distribution function. The last line simply follows from $z_1\independent z_2$ and $z_{1,2}\sim\mathcal{N}\lr0,1\rr$. Similarly, by arranging the last three lines of Eq. \eqref{eqP1x1ac} and changing the direction of inequality from $>$ to $<$, we can indeed verify that
\begin{equation}
\label{eqP1x0ac}
P_t^1(\xi_c|x_0,s) = P_t^1(\xi_c|x_1,s)
\end{equation}
and, moreover, by virtue of convex combination in Eq. \eqref{eqXicXie1}, that
\begin{equation}
\label{eqP1xics}
P_t^1(\xi_c|s) = P_t^1(\xi_c|x_1,s)=P_t^1(\xi_c|x_0,s).
\end{equation}
Equations \eqref{eqP1x0ac} and \eqref{eqP1x1ac} then directly imply
\begin{equation}
\label{eqP2xics}
P_t^2(\xi_c|s) = 1-P_t^1(\xi_c|s) = P_t^1(\xi_e|s).
\end{equation}
Thus, the chances of agent $j$ getting a correct (or erroneous) signal are the same no matter if the market is bullish or bearish, and one agent's success is the other one's failure, as expected. Note also that $P_{t}^1(\xi_c |x_1,s)$ in Eq. \eqref{eqP1x1ac} is not a function of the value of agent $j$'s specific information at time $t$, i.e., $\bar{\sigma}(\xi_t^j)$, but rather depends only on the differential between information flow speeds, $\sigma_1$ and $\sigma_2$ (or, how agents perceive it), and the spread of $X$.

Eq. \eqref{eqP1x1ac} indeed reveals a number of intuitive properties, such as: (i) the larger the differential $|\sigma_1-\sigma_2|$, the more likely the agent with superior signal will get a correct signal ($\xi_c$), (ii) with $|\sigma_1-\sigma_2|$ given, the agent with superior signal will prefer more uncertainty (i.e., greater spread for $X$) to less uncertainty (i.e., smaller spread for $X$), and, (iii) with $|\sigma_1-\sigma_2|$  and spread of $X$ given, refraining from a trade will always result in greater chances of getting a correct signal (although there will be a cost to refraining).

To complete the case where the contract pays a binary dividend, we state, by virtue of Eq. \eqref{eqEOverallProfit}, the expected `profit-to-go' of agent $j$ at time $t$:
\begin{eqnarray}
\label{eqEOverallProfit2go}
\displaystyle\sum_{u=t}^T \E_t^j \ls \Pi_{uT}^j\rs &=& \displaystyle\sum_{u=t}^T P_t^j (x_1)\lr P_u^j(\xi_c|x_1,s_u) \lr x_1 -\E_t^j \ls S_u^\ast (x_1) | \xi_c\rs\rr \right.\nn\\
&&\left. - P_u^j(\xi_e|x_1) \lr x_1 -\E_t^j \ls S_u^\ast (x_1) | \xi_e\rs\rr \rr + \nn\\
&& \displaystyle\sum_{u=t}^T P_t^j (x_0)\lr P_u^j(\xi_c|x_0) \lr \E_t^j \ls S_u^\ast (x_0) | \xi_c \rs - x_0 \rr \right.\nn\\
&&\left. - P_u^j(\xi_e |x_0) \lr \E_t^j \ls S_u^\ast (x_0) | \xi_e \rs - x_0\rr \rr.
\end{eqnarray}
Note that we preserve the subscript $t$ for $p^j$ and $\E^j$ as they will be inferred based on the effective information at time $t$, i.e., $\bar{\sigma}(\xi_t^j)$.  Below we generalise the above results to the case where $X$ has a continuous distribution.

\subsubsection{Directional Quality of Trading Signal (Gaussian Dividend)}

We first redefine the likelihoods of high- and low-type markets (see Eq. \eqref{eqHighLowCont}):
\begin{equation}
P_t^j (x^+) = \displaystyle\int_{\mathbb{X}^+}\pi_t^j(x)\text{d}x,\quad P_t^j (x^-) = \displaystyle\int_{\mathbb{X}^-}\pi_t^j (x)\text{d}x,
\end{equation}
where $\mathbb{X}^+=(0,\infty)$ and $\mathbb{X}^-=(-\infty,0)$. When $s_t=s$ is non-zero, Eq. \eqref{eqValueGauss} takes the form 
\begin{equation}
S_t^{1}=\mathds{1}_{\left\{t<T\right\}} e^{-r(T-t)}\frac{\sigma_{1} \kappa_t \xi_{t}^{1}+\sigma_{2} \kappa_s \xi_{s}^{2}}{\sigma_{1}^2 \kappa_t t+\sigma_{2}^2 \kappa_s s+1},
\end{equation}
Then, the chances for agent $1$ having right trade signals in high- and low-type markets can be found as follows. Let $X=x$, $x>0$ (high-type). Then,
\begin{eqnarray}
\label{eqP1xacg}
P_t^1(\xi_c|x^+,s) &=&\displaystyle\int_{\mathbb{X}^+} P^1(S_t^1 (x)>S_t^\ast (x)|s) \bar{\pi}_t^{1+} (x)\text{d}x \nn\\
&=& \displaystyle\int_{\mathbb{X}^+} P^1 (S_t^1 (x) / 2>S_t^2 (x) / 2 |s) \bar{\pi}_t^{1+} (x)\text{d}x \nn\\
&=&\displaystyle\int_{\mathbb{X}^+} P\lr\frac{1}{2}\frac{\sigma_1\kappa_t \xi_t^1 + \sigma_2\kappa_s \xi_s^2}{\sigma_1^2 t \kappa_t+\sigma_2^2 s \kappa_s+1}  - \frac{1}{2}\frac{\sigma_2\kappa_t \xi_t^2 + \sigma_1\kappa_s \xi_s^1}{\sigma_2^2 t \kappa_t+\sigma_1^2 s \kappa_s+1}>0\rr \bar{\pi}_t^{1+} (x)\text{d}x \nn\\
&=&\displaystyle\int_{\mathbb{X}^+} P\lr\frac{1}{2}\frac{\sigma_1\kappa_t (\sigma_1 t x+\beta_t^1)}{\sigma_1^2 t \kappa_t+\sigma_2^2 s \kappa_s+1}  - \frac{1}{2}\frac{\sigma_1\kappa_s (\sigma_1 s x+\beta_s^1)}{\sigma_1^2 s \kappa_s+\sigma_2^2 t \kappa_t+1}\right.\nn\\
&&\left.-\lr\frac{1}{2}\frac{\sigma_2\kappa_t (\sigma_2 t x+\beta_t^2)}{\sigma_2^2 t \kappa_t+\sigma_1^2 s \kappa_s+1}-\frac{1}{2}\frac{\sigma_2\kappa_s (\sigma_2 s x+\beta_s^2)}{\sigma_2^2 s \kappa_s+\sigma_1^2 t \kappa_t+1}\rr>0\rr \bar{\pi}_t^{1+} (x)\text{d}x \nn\\
&=&\displaystyle\int_{\mathbb{X}^+} \Theta\lr-\frac{a_t^s}{b_t^s}x\rr\bar{\pi}_t^{1+} (x)\text{d}x\quad(a_t^s, b_t^s, x >0),
\end{eqnarray}
where, again, $z_{1,2}$ are independently $\mathcal{N}(0,1)$, $\bar{\pi}^+$ is the normalised effective posterior density as given in Eq. \eqref{eqJointP},
\begin{eqnarray}
a_t^s &=& \frac{1}{2}\lr\frac{\sigma_2^2 t \kappa_t}{\sigma_2^2 t \kappa_t+\sigma_1^2 s \kappa_s+1}-\frac{\sigma_2^2 s\kappa_s}{\sigma_2^2 s \kappa_s+\sigma_1^2 t \kappa_t+1}\rr\nn\\
&&-\frac{1}{2}\lr\frac{\sigma_1^2 t \kappa_t}{\sigma_1^2 t \kappa_t+\sigma_2^2 s \kappa_s+1}  - \frac{\sigma_1^2 s \kappa_s}{\sigma_1^2 s \kappa_s+\sigma_2^2 t \kappa_t+1}\rr,
\end{eqnarray}
\begin{eqnarray}
b_t^s &=& \frac{1}{2}\ls\lr\frac{\sigma_1\kappa_t}{\sigma_1^2 t \kappa_t+\sigma_2^2 s \kappa_s+1}\rr^2 (t/\kappa_t)  + \lr\frac{\sigma_1\kappa_s}{\sigma_1^2 s \kappa_s+\sigma_2^2 t \kappa_t+1}\rr^2 (s/\kappa_s)\right.\nn\\
&&\left. -2\lr\frac{\sigma_1\kappa_t}{\sigma_1^2 t \kappa_t+\sigma_2^2 s \kappa_s+1}\rr\lr\frac{\sigma_1\kappa_s}{\sigma_1^2 s \kappa_s+\sigma_2^2 t \kappa_t+1}\rr(s/\kappa_t)\right.\nn\\
&&\left.+\lr\frac{\sigma_2\kappa_t}{\sigma_2^2 t \kappa_t+\sigma_1^2 s \kappa_s+1}\rr^2 (t/\kappa_t) + \lr\frac{\sigma_2\kappa_s}{\sigma_2^2 s \kappa_s+\sigma_1^2 t \kappa_t+1}\rr^2 (s/\kappa_s) \right.\nn\\
&&\left. -2\lr\frac{\sigma_2\kappa_t}{\sigma_2^2 t \kappa_t+\sigma_1^2 s \kappa_s+1}\rr\lr\frac{\sigma_2\kappa_s}{\sigma_2^2 s \kappa_s+\sigma_1^2 t \kappa_t+1}\rr(s/\kappa_t)\rs^\frac{1}{2}.
\end{eqnarray}

Indeed, we can quickly verify that $a_t^s>0$ (i.e., $x>y$ implies $x/(x+1)>y/(y+1)$. Moreover, similar to the digital payoff case, it directly follows from \eqref{eqP1xacg} that $P_t^1(\xi_c|x^+,s)=P_t^1(\xi_c|x^-,s)$, $P_t^1(\xi_c|s)=P_t^1(\xi_c|x^+,s)=P_t^1(\xi_c|x^-,s)$, and $P_t^2(\xi_c|s) = 1-P_t^1(\xi_c|s) = P_t^1(\xi_e|s)$. 

For a low-type market, similarly,
\begin{eqnarray}
P_t^1(\xi_c|x^-,s)&=&\displaystyle\int_{\mathbb{X}^-}\Theta\lr \frac{a_t^s}{b_t^s}x\rr \pi_t^- (x)\text{d}x\quad(a_t^s,b_t^s>0,x<0)\nn\\
&=&\displaystyle\int_{\mathbb{X}^+}\Theta\lr -\frac{a_t^s}{b_t^s}x\rr \pi_t^+ (x)\text{d}x\quad(a_t^s,b_t^s>0,x>0)\nn\\
&=& P_t^1(\xi_c|x^+,s).
\end{eqnarray}

Expected profit-to-go at time $t$ can be inferred in a similar sense to Eq. \eqref{eqEOverallProfit}. As a final step to calculate the signal-based expected profit of then agent at time $t$ (i.e., just before the auction), as given in Eq. \eqref{eqOverallProfitre3}, we now compute the expected transaction price in low- and high-type markets and with correct and erroneous signals. Note that 
\begin{eqnarray}
\label{eqPDifference1}
\E_t^1 \ls S_t^\ast (x) \bigr\rvert x^+,\xi_c,s\rs &=& \sfrac{1}{2}\lr S_t^1(x) + \E_t^1 \ls S_t^2 (x) \bigr\rvert x^+,\xi_c,s \rs\rr \nn\\
&=& S_t^1(x) - \E \ls \sfrac{1}{2}\lr S_t^1 (x) - S_t^2 (x)\rr^+ \bigr\rvert x^+,s\rs,
\end{eqnarray}
\begin{eqnarray}
\label{eqPDifference2}
\E_t^1 \ls S_t^\ast (x) \bigr\rvert x^+,\xi_e,s\rs &=& \sfrac{1}{2}\lr S_t^1(x) + \E_t^1 \ls S_t^2 (x) \bigr\rvert x^+,\xi_e,s \rs\rr \nn\\
&=& S_t^1(x) + \E \ls \sfrac{1}{2}\lr S_t^2 (x) - S_t^1 (x)\rr^+ \bigr\rvert x^+,s\rs,
\end{eqnarray}
\begin{eqnarray}
\label{eqPDifference3}
\E_t^1 \ls S_t^\ast (x) \bigr\rvert x^-,\xi_c\rs &=& \sfrac{1}{2}\lr S_t^1(x) + \E_t^1 \ls S_t^2 (x) \bigr\rvert x^-,\xi_c,s\rs\rr \nn\\
&=& S_t^1(x) + \E \ls \sfrac{1}{2}\lr S_t^2 (x) - S_t^1 (x)\rr^+ \bigr\rvert x^-,s \rs
\end{eqnarray}
and
\begin{eqnarray}
\label{eqPDifference4}
\E_t^1 \ls S_t^\ast (x) \bigr\rvert x^-,\xi_e \rs &=& \sfrac{1}{2}\lr S_t^1(x) + \E_t^1 \ls S_t^2 (x) \bigr\rvert x^-,\xi_e,s\rs\rr \nn\\
&=& S_t^1(x) - \E \ls \sfrac{1}{2}\lr S_t^1 (x) - S_t^2 (x)\rr^+ \bigr\rvert x^-,s \rs.
\end{eqnarray}

Notice that we dropped $t$ and $j$ from $\E_t^j$'s in Eq. \eqref{eqPDifference1}-\eqref{eqPDifference4} since, by Eq. \eqref{eqP1xacg}, when $(\xi_t^j)_{0\leq t\leq T}$ is pinned to a certain value $x$, the price differential is not conditional on the specific value of agent $j$'s signal at time $t$, but rather a function of $\sigma_1$, $\sigma_2$, $t$ and $s$. Thus, all one needs to do (so as to compute the expected transaction price under different market situations and trading signal quality) is to work out the expected value of the `absolute price differential' under each circumstance. To that end, we can infer from Eq. \eqref{eqP1xacg} that 
\begin{eqnarray}
&\lr\sfrac{1}{2}\lr S_t^2 (x)-S_t^1 (x)\rr \bigr\rvert x^+,s\rr \sim \N(a_t^s x,b_t^s),&\nn\\
&\lr\sfrac{1}{2}\lr S_t^2 (x)-S_t^1 (x)\rr  \bigr\rvert x^-,s \rr \sim \N(a_t^s x,b_t^s),&\nn\\
&\lr\sfrac{1}{2}\lr S_t^1 (x)-S_t^2 (x)\rr \bigr\rvert x^+,s\rr \sim \N(-a_t^s x,b_t^s),&\nn\\
&\lr\sfrac{1}{2}\lr S_t^1 (x)-S_t^2 (x)\rr \bigr\rvert x^-,s\rr \sim \N(-a_t^s x,b_t^s),&\nn
\end{eqnarray}
where $a_t^s,b_t^s>0$ are as given above. As a result,
\begin{eqnarray}
\label{eqPDifferenceA}
\E\ls \sfrac{1}{2} \lr S_t^2 (x)-S_t^1 (x)\rr^+\bigr\rvert x^+,s\rs &=&\left\{ \begin{array}{ll}
      \eta_1\frac{1}{\sqrt{2\pi}b_t^s}\displaystyle\int_{0}^{\infty} y e^{-\frac{1}{2}\frac{(y-a_t^s x)^2}{(b_t^s)^2}}\text{d}y, & x>0, \\
      \eta_1\frac{1}{\sqrt{2\pi}b_t^s}\displaystyle\int_{0}^{\infty} y e^{-\frac{1}{2}\frac{(y+a_t^s x)^2}{(b_t^s)^2}}\text{d}y, & x<0,
\end{array} 
\right.\nn\\
&=& \E \ls \sfrac{1}{2} \lr S_t^1 (x)-S_t^2 (x)\rr^+ \bigr\rvert x^-,s\rs, \nn\\
\end{eqnarray}
and
\begin{eqnarray}
\label{eqPDifferenceB}
\E\ls \sfrac{1}{2} \lr S_t^1 (x)-S_t^2 (x) \rr^+ \bigr\rvert x^+,s \rs &=&\left\{ \begin{array}{ll}
      \eta_2\frac{1}{\sqrt{2\pi}b_t^s}\displaystyle\int_{0}^{\infty} y e^{-\frac{1}{2}\frac{(y + a_t^s x)^2}{(b_t^s)^2}}\text{d}y, & x>0, \\
      \eta_2\frac{1}{\sqrt{2\pi}b_t^s}\displaystyle\int_{0}^{\infty} y e^{-\frac{1}{2}\frac{(y - a_t^s x)^2}{(b_t^s)^2}}\text{d}y, & x<0,
\end{array} 
\right.\nn\\
&=& \E \ls \sfrac{1}{2} \lr S_t^2 (x)-S_t^1 (x)\rr^+ \bigr\rvert x^-,s\rs, \nn\\
\end{eqnarray}
with density normalising factors $\eta_1$, $\eta_2$. Thus, under the Gaussian payoff scenario, we have derived explicit formulae for the two main sources of uncertainty involved in signal-based trade, namely, the likelihood of a signal's pointing at the right (wrong) trade direction, and the expected amount of profit (loss) given the signal was correct (erroneous). With inputs from Eq. \eqref{eqP1xacg} and \eqref{eqPDifference1}-\eqref{eqPDifference4}, Eq. \eqref{eqOverallProfitre3} can now be written for agent $1$ as
\begin{eqnarray}
\label{eqOverallProfitre2}
\E_t^1\ls \Pi_t^{1} \rs &=&P_t^1 (x^+) \lr \displaystyle\int_{\mathbb{X}^+} \lr\Theta\lr-\frac{a_t^s}{b_t^s}x\rr \eta_1 \frac{1}{\sqrt{2\pi}b_t^s}\displaystyle\int_{0}^{\infty} y e^{-\frac{1}{2}\frac{(y+a_t^s x)^2}{(b_t^s)^2}}\text{d}y \right.\right. \nn\\
&& \left.\left. - \Theta\lr\frac{a_t^s}{b_t^s}x\rr \eta_2 \frac{1}{\sqrt{2\pi}b_t^s}\displaystyle\int_{0}^{\infty} y e^{-\frac{1}{2}\frac{(y - a_t^s x)^2}{(b_t^s)^2}}\text{d}y \rr \bar{\pi}_t^{1+}(x) \text{d}x\rr + \nn\\
&& P_t^1 (x^-) \lr  \displaystyle\int_{\mathbb{X}^-} \lr \Theta\lr\frac{a_t^s}{b_t^s}x\rr \eta_1 \frac{1}{\sqrt{2\pi}b_t^s}\displaystyle\int_{0}^{\infty} y e^{-\frac{1}{2}\frac{(y-a_t^s x)^2}{(b_t^s)^2}}\text{d}y \right. \right. \nn\\
&& \left. \left.- \Theta\lr-\frac{a_t^s}{b_t^s}x\rr \eta_2 \frac{1}{\sqrt{2\pi}b_t^s}\displaystyle\int_{0}^{\infty} y e^{-\frac{1}{2}\frac{(y+a_t^s x)^2}{(b_t^s)^2}}\text{d}y \rr \bar{\pi}_t^{1-}(x) \text{d}x\rr, \nn\\
\end{eqnarray}
where $\bar{\pi}^{+}$, $\bar{\pi}^{-}$ are, again, normalised effective posteriors for high- and low-type markets, with $\mathbb{X}^+=(0,\infty)$ and $\mathbb{X}^-=(-\infty,0)$. Equation \eqref{eqOverallProfitre2} can also be written for agent $2$ without much effort. Accordingly, agent $j$ updates his trading procedure as follows: (1)-(1a) Same as in Section \ref{secNumerical} above. (1b) Calculate $\E_t^j [\Pi_t^j]$ based on Eq. \eqref{eqOverallProfitre2}. (1c) Decide whether to quote or not to quote a price. If yes, proceed to next step. (2) Quote signal-based price (as $\varsigma^\pm=1$). (3) Same as in Section \ref{secNumerical} above. 

Now, equipped with the flexibility to shape his strategy $(q_t^j)_{0\leq t\leq T}$, $|q_t^j|\in\{0,1\}$, by timing his trades, agent $j$ will need to develop an optimal `online' trading rule (referring to (1c) above) that maximises his profits, in understanding of his marginal benefits and losses from seizing or skipping a trade opportunity.

\subsection{Risk-neutral Optimal Strategy}

It is not difficult to see that $\E_t^1\ls \Pi_t^{1} \rs$ in Eq. \eqref{eqOverallProfitre2} will always be negative when agent knows that his information is less superior. The top and bottom left panels of Figure \ref{figAllPossibleStrategiesEPi12_2sigma_x05_c}, in this regard, depicts the evolution in time of $\E_t^1 \ls \Pi_t^1 \rs$, based on almost all possible strategies and a sample path of $\xi_t^1$, when an agent believes that he is informationally less susceptible than his counterpart. Yet, the agent can minimise his chances of losing from a trade by keeping his information up-to-date through trading at ``each'' time step (i.e., the top edge of each shape in the left panel). We note that the marginal expected cost of refraining from trade for the less susceptible agent is always positive when the agent believes he is informationally less susceptible.

We therefore infer from Figure \ref{figAllPossibleStrategiesEPi12_2sigma_x05_c} that a solution to the maximisation problem in Eq. \eqref{eqMaxRiskNeutral} is unattainable from the perspective of a less informationally capable agent. The real-world implication of this is that market shutdowns may not occur in a real market setting because investors think their effective information are either constantly or temporarily superior to the market information. We, thus, turn our focus to the case where both agents believe their information source is characterised by a higher $\sigma$.

The top and bottom right panels of Figure \ref{figAllPossibleStrategiesEPi12_2sigma_x05_c}, on the other hand, shows the evolution in time of $\E_t^j \ls \Pi_t^j \rs$ for the agent who believes that he is informationally more susceptible. The strategy which results in the bottom edge these shapes on the right panel is unique, i.e., $|q_t^j|=1$ $\forall t$. However, there is no single strategy which can achieve the top edge of each shape, which is a combination of different strategies that result in the maximum expected potential at different time points.

\begin{figure}[h]
\centering
\includegraphics[scale=0.50]{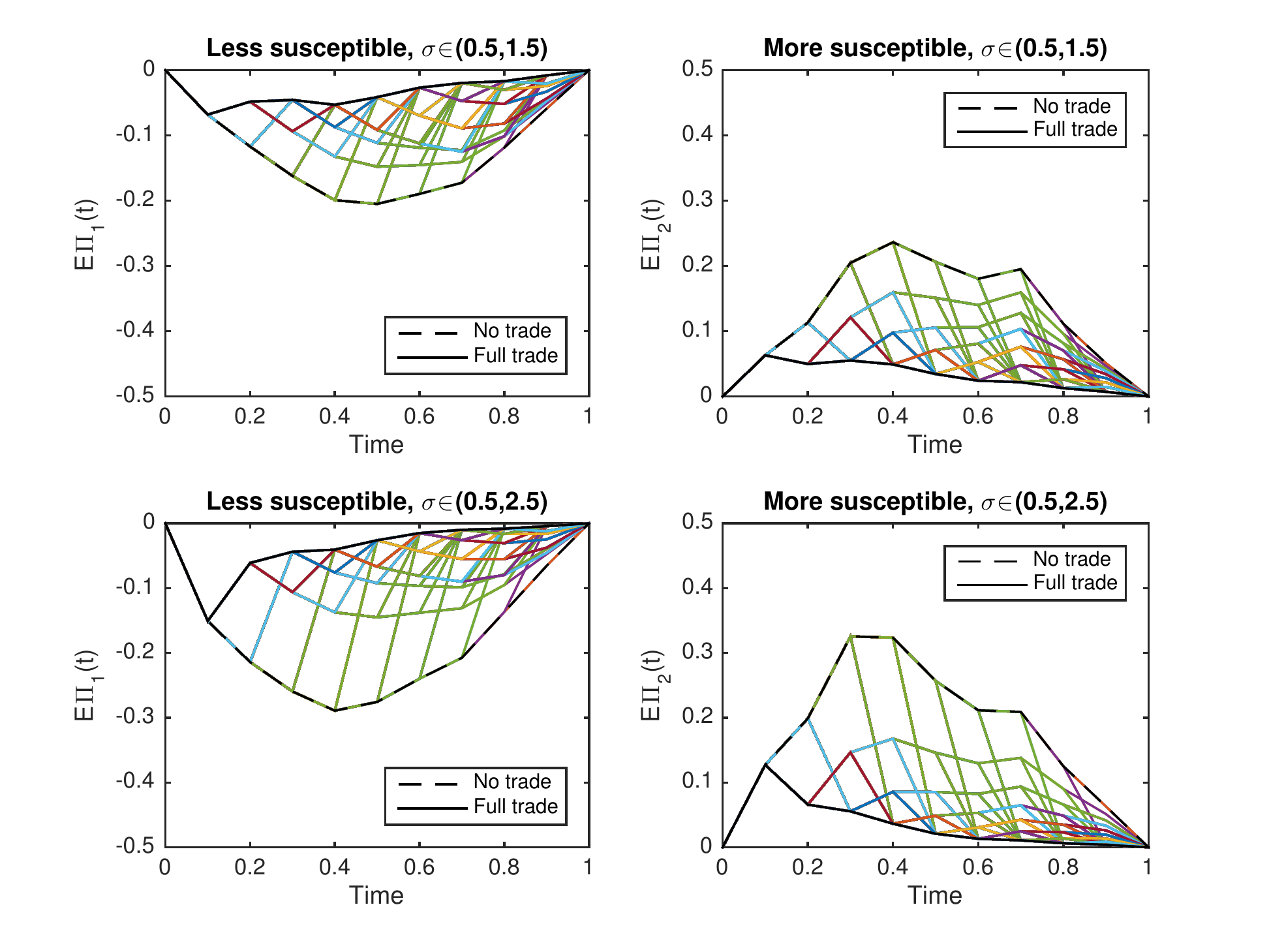}
\caption{Evolution of $\E_t^j [ \Pi_t^j ]$ as given in Eq. \eqref{eqOverallProfitre2} for sample trajectories of $\xi_t^1$ and $\xi_t^2$ and all possible trading strategies. The dividend is assumed to be Gaussian.} 
\label{figAllPossibleStrategiesEPi12_2sigma_x05_c}
\end{figure}

We define the optimal strategy of an agent as the one which maximises his overall expected terminal profit from trading the contract based on his effective information $\bar{\sigma}(\xi_t^j)$, i.e., for agent 1,
\begin{eqnarray}
\label{eqMaxRiskNeutral}
&\arg\max_{(q_t^1)} &\displaystyle\sum_{t=0}^T\E_t^1\ls \Pi_t^{1}\rs =  \displaystyle\sum_{t=0}^T\E \ls \Pi_t^{1} | \bar{\sigma}(\xi_t^1) \rs \nn\\ 
&\text{s.t.}& S_t^\ast = \sfrac{1}{2}(S_t^1+S_t^2), \\
&&q_t=0\nn\\
&&\forall t, \nn
\end{eqnarray}
where $\bar{\sigma}(\xi_t^1)$ is same as in Eq. \eqref{eqEffectiveInfo}.

On an extra note, we remark that setting a mid-price, as in Eq. \eqref{eqMidPrice} and \eqref{eqMaxRiskNeutral}, is indeed equivalent to
\begin{equation}
\E\ls q_t^1 (X-S_t^\ast) | \bar{\sigma}(\xi_t^1) \rs = \E\ls q_t^{2} (X-S_t^\ast) | \bar{\sigma}(\xi_t^2) \rs .
\end{equation}

Thus, we can reinterpret the role of the central planner, in the context of this section, as `to observe $\bar{\xi}_t^j$'s through price quotes and set the transaction price as the mid-price which equates the signal-based terminal profits of agents.'

Similar to \cite{BussDumas2013}, we can define the dynamic programming formulation of the agent $j$'s problem given in Eq. \eqref{eqMaxRiskNeutral} as follows:
\begin{equation}
\label{eqValueFunction}
V_t^j = \sup_{(q_t^j)} \lr \E_t^{j}\ls \Pi_t^{j} \rs + \E_t^j \ls V_{t+1}^j \rs \rr,
\end{equation}
where $V_t^j$ is the value function. Note that $\E_t^{j}[ \Pi_t^{j} ]$ is implicitly determined by $(q_s^j)_{0\leq s <t}$, whereas $\E_t^j [ V_{t+1}^j ]$ by $(q_s^j)_{0\leq s \leq t}$. Therefore, at each auction, the agent will need to consider the marginal impact of his current strategy on $\E_t^j [ V_{t+1}^j ]$. The particular nature of the present model, however, does not allow us to employ a backward-induction technique that is similar to the one described in \cite{BussDumas2013} and \cite{DumasLyasoff2012}.

Based on Eq. \eqref{eqValueFunction}, we introduce the following real-time optimal trading strategy for agent $j$:
\begin{equation}
\label{eqDecision}
|q_t^j| = \left\{ \begin{array}{ll}
      1, &\text{if}\quad \E_t^{j}\ls \Pi_t^{j} \rs >0, \\
	  &\text{and}\quad\E_t^{j}\ls \Pi_t^{j} \rs + \E_t^{j}\ls \Pi_{t+1}^{j}\rs_{|q_t^j| =1} > \E_t^{j}\ls \Pi_{t+1}^{j} \rs_{q_t^j =0}, \\          
      0, & \text{otherwise},
\end{array} 
\right.
\end{equation}
where the term
\begin{equation}
\label{eqDecision2}
\E_t^{j}[ \Pi_t^{j} ] - \lr \E_t^{j}[\Pi_{t+1}^{j}]_{q_t^j =0} - \E_t^{j}[\Pi_{t+1}^{j}]_{|q_t^j| =1} \rr
\end{equation}
can be seen as the immediate expected gain from trade adjusted for the cost of losing the informational advantage. Thus, the agent chooses to trade whenever his cost-adjusted expected gain from trade is strictly positive. In other words, whenever an agent refrains from trade in expectation of greater future profits, he should do so on the basis that he has to recover immediate cost of refraining.

Figure \ref{figDecisionRuleValueEPi2_sigma0510_c} plots the value of \eqref{eqDecision2} (averaged over a number of sample paths of $\xi_t^j$) for the more informationally susceptible agent for each point $t$ in the trading horizon and given each possible trading history $s_t$. The decision rule variable is positive for any possible past strategy characterised by the last time of trade, $q_{s_t}^j$, implying that the agent can maximise the sum of his expected terminal profits by trading at each time point $t$, thereby constantly incorporating his differential information into prices, i.e.,  $|q_t^j|=1$ $\forall t\in[0,T]$.

\begin{figure}[h]
\centering
\includegraphics[scale=0.40]{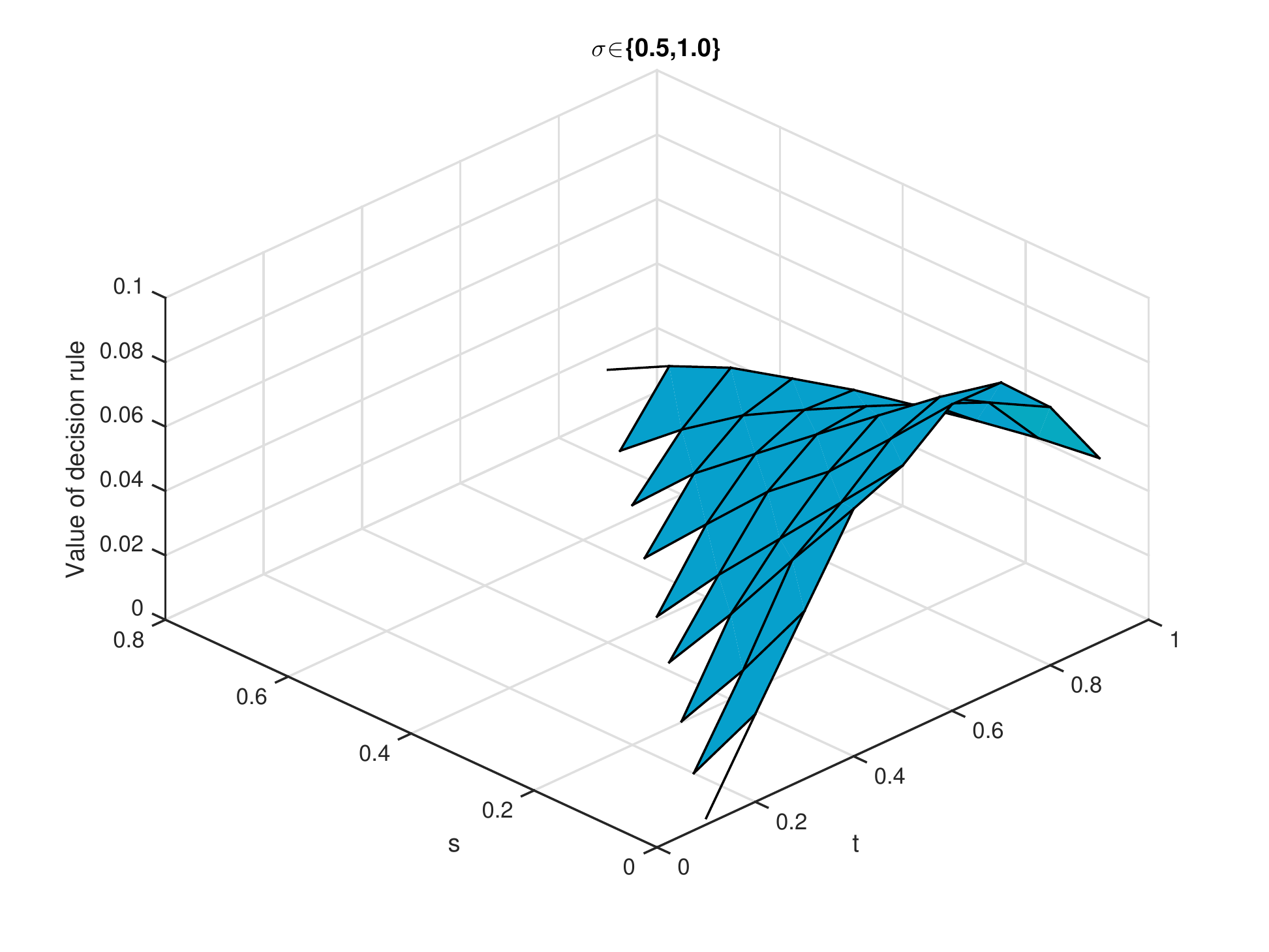}
\caption{Value of cost-adjusted expected gain from trade (averaged over a number of sample paths of $\xi_t^j$) for the more informationally susceptible agent based on Eq. \eqref{eqDecision} for all $t$ and $s_t$.} 
\label{figDecisionRuleValueEPi2_sigma0510_c}
\end{figure}

We can indeed show that the optimality of this strategy is invariant to the path of $\xi_t^j$. Consider agent $2$ who deems his signal superior ($\sigma_2>\sigma_1$) and let the market be high-type (i.e., $x\in\mathbb{X}^+$). For given $x$, $\sigma_1$, $\sigma_2$, let's denote the corresponding integrand in Eq. \eqref{eqOverallProfitre2}, rearranged for agent $2$, by $H_2(t,s)$, where
\begin{eqnarray}
\label{eqIntegrand}
H_2 (t,s)&=&\Theta\lr\frac{a_t^s}{b_t^s}x\rr \eta_1 \frac{1}{\sqrt{2\pi}b_t^s}\displaystyle\int_{0}^{\infty} y e^{-\frac{1}{2}\frac{(y-a_t^s x)^2}{(b_t^s)^2}}\text{d}y\nn\\
&&-\Theta\lr-\frac{a_t^s}{b_t^s}x\rr \eta_1 \frac{1}{\sqrt{2\pi}b_t^s}\displaystyle\int_{0}^{\infty} y e^{-\frac{1}{2}\frac{(y+a_t^s x)^2}{(b_t^s)^2}}\text{d}y.
\end{eqnarray}

Note that $H_2(t,s)$ in Eq. \eqref{eqIntegrand} is the expected profit of agent 2 at time $t$ given the time of last trade, $s$, and a high-type payoff $x$, and it is not a function of $\xi_t^2$. Agent 2 version of Eq. \eqref{eqOverallProfitre2} is, in fact, nothing but the sum of convex combinations of $H_2(t,s)$ and its low-market analogous $L_2(t,s)$ with respect to the effective posteriors $\bar{\pi}_t^{2+}$ and $\bar{\pi}_t^{2-}$, respectively. Thus, similar to the relation \eqref{eqDecision2},
\begin{equation}
H_2(t,s)-\ls H_2(t+1,s) - H_2(t+1,t) \rs
\end{equation}
can be seen as the signal-independent version of the cost-adjusted immediate gain from trade (for the agent who deems his signal superior), whose value is depicted in Figure \ref{figDecisionRuleValueEPi2_sigma0510_c2}. It can be inferred from the figure, again, that it is optimal for the informationally more susceptible agent to trade continuously without accumulating his extra information.

\begin{figure}[h]
\centering
\includegraphics[scale=0.40]{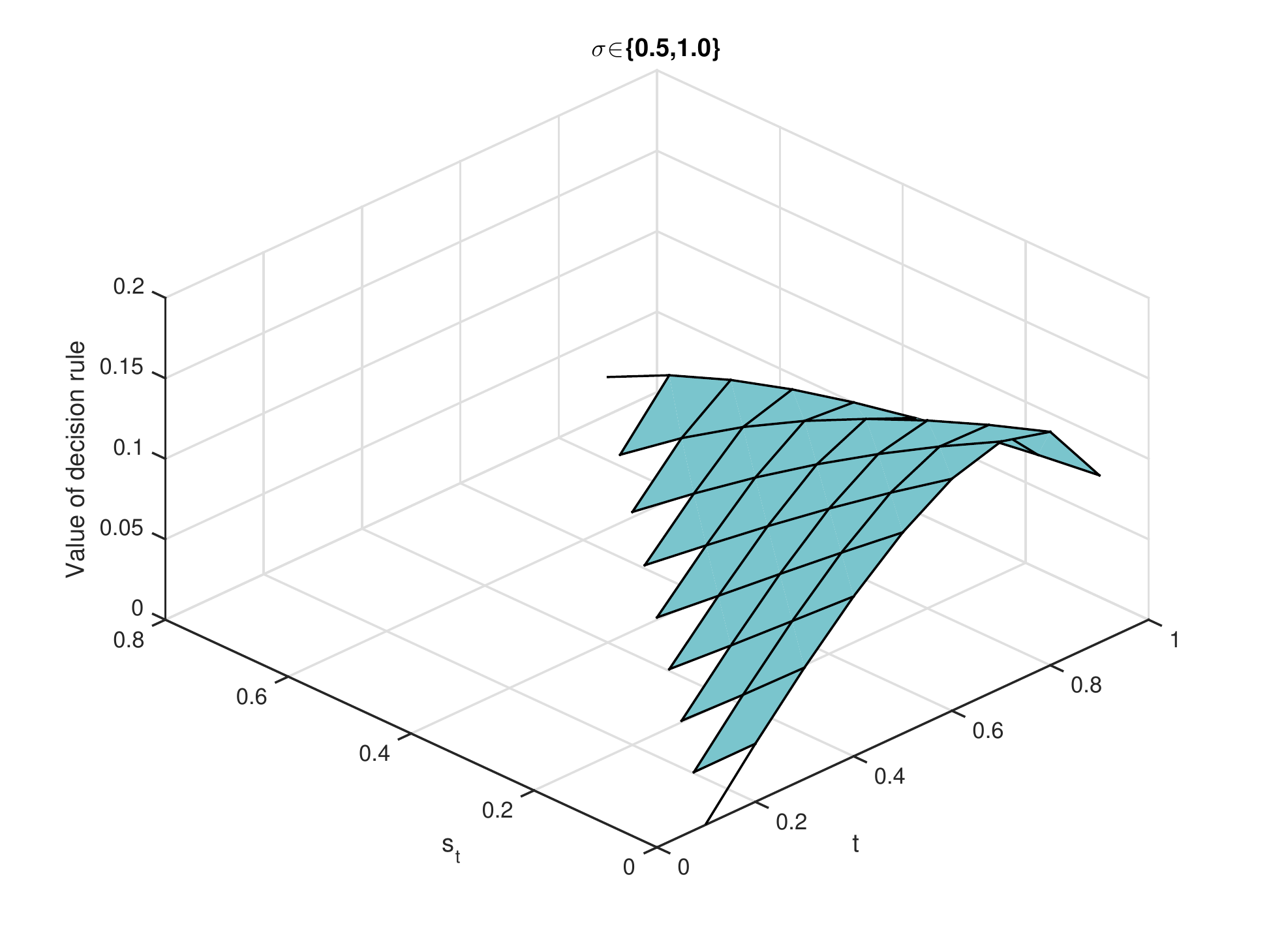}
\caption{Value of signal-independent cost-adjusted expected gain from trade for the more informationally susceptible agent based on Eq. \eqref{eqDecision} for all $t$ and $s_t$.} 
\label{figDecisionRuleValueEPi2_sigma0510_c2}
\end{figure}


\subsection{Extension to Risk-adjusted Performance}

In case agents are risk-adjusted expected profit (e.g., Sharpe ratio, \cite{Sharpe1966}) maximisers ``at the portfolio level'', the objective function in Eq. \eqref{eqMaxRiskNeutral} can simply be modified as
\begin{equation}
\label{eqMaxRiskAdjusted}
\arg\max_{(q_t^j)} \frac{\displaystyle\sum_{t=0}^T q_t^j \E_t^j\ls \Pi_t^{j}\rs}{\lr \displaystyle\sum_{t=0}^T (q_t^j)^2 \V_t^j\lr \Pi_t^{j}\rr \rr^{\sfrac{1}{2}}}. 
\end{equation}

We then write the conditional variance $\V_t^j(\Pi_t^{j})=\V^j(\Pi_t^{j}|\bar{\sigma}(\xi_t^j))$ of the signal-based profit at time $t$, whose expectation is given in Eq. \eqref{eqOverallProfitre2} as
\begin{eqnarray}
\label{eqOverallProfitVarre2}
\V_t^1( \Pi_t^{1} ) &=&p_t^1 (x^+) \lr \displaystyle\int_{\mathbb{X}^+} \lr\Theta\lr-\frac{a_t^s}{b_t^s}x\rr^2 \eta_1 \frac{1}{\sqrt{2\pi}b_t^s}\displaystyle\int_{0}^{\infty} y^2 e^{-\frac{1}{2}\frac{(y+a_t^s x)^2}{(b_t^s)^2}}\text{d}y \right.\right. \nn\\
&& \left.\left. - \Theta\lr\frac{a_t^s}{b_t^s}x\rr^2 \eta_2 \frac{1}{\sqrt{2\pi}b_t^s}\displaystyle\int_{0}^{\infty} y^2 e^{-\frac{1}{2}\frac{(y - a_t^s x)^2}{(b_t^s)^2}}\text{d}y \rr \bar{\pi}_t^{1+}(x) \text{d}x\rr + \nn\\
&& p_t^1 (x^-) \lr  \displaystyle\int_{\mathbb{X}^-} \lr \Theta\lr\frac{a_t^s}{b_t^s}x\rr^2 \eta_1 \frac{1}{\sqrt{2\pi}b_t^s}\displaystyle\int_{0}^{\infty} y^2 e^{-\frac{1}{2}\frac{(y-a_t^s x)^2}{(b_t^s)^2}}\text{d}y \right. \right. \nn\\
&& \left. \left.- \Theta\lr-\frac{a_t^s}{b_t^s}x\rr^2 \eta_2 \frac{1}{\sqrt{2\pi}b_t^s}\displaystyle\int_{0}^{\infty} y^2 e^{-\frac{1}{2}\frac{(y+a_t^s x)^2}{(b_t^s)^2}}\text{d}y \rr \bar{\pi}_t^{1-}(x) \text{d}x \rr \nn\\
&& - \lr\E_t^1[\Pi_t^{1}]\rr^2. 
\end{eqnarray}

Figure \ref{figAllPossibleStrategiesSHRP12_2sigma_x05_c} depicts the risk-adjusted version of Figure \ref{figAllPossibleStrategiesEPi12_2sigma_x05_c} using the same signal sample as in the latter.

\begin{figure}[h]
\centering
\includegraphics[scale=0.50]{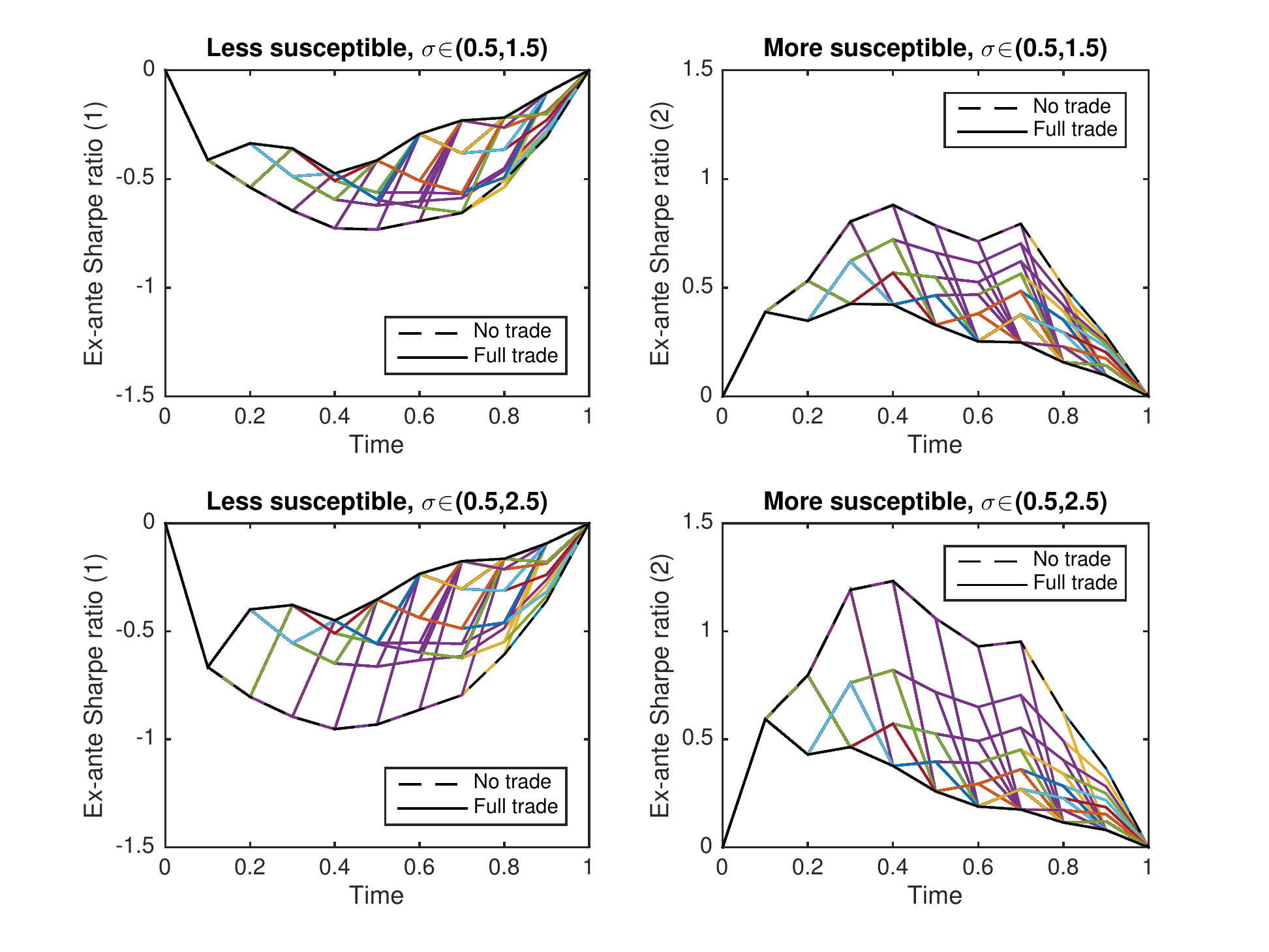}
\caption{Evolution of Sharpe ratio based on $\V_t^j ( \Pi_t^j )$ given in Eq. \eqref{eqOverallProfitVarre2} for sample trajectories of $\xi_t^1$ and $\xi_t^2$ and all possible trading strategies. The dividend is assumed to be Gaussian.} 
\label{figAllPossibleStrategiesSHRP12_2sigma_x05_c}
\end{figure}

\subsection{Extension to Risk-averse Utility}

The above setup can easily be generalised to the case where agents are `characteristically' risk-averse and attach decreasing marginal utility to each extra unit of expected return due to the additional risk involved. In \cite{BondEraslan2010}, the authors show that the following two cases are equivalent: (a) terminal payoff is exogenous (as in our case) and agents are risk-averse, (b) dividend is endogenous and agents are risk-neutral. When the asset dividend (or terminal payoff) is exogenous and agents' actions have no impact on it, one needs to introduce either trade quotas or proper risk aversion assumptions to prevent agents from trading unlimited amounts to make infinite profits, should quoted prices be in their favour. In the presence of informational differences, there would be less or no motivation for agents who are not only informationally less capable but also risk averse to actively participate in a market where the participants are assumed to be rational. Such state of affairs can, in fact, exacerbate the situations where markets shut down due to perceived differential information. Such situations are avoided in the literature by introducing the concept of `noise-traders', which we avoid in the present context so as to focus solely on the influence of differential information on market phenomena.

We assume that agents are risk-averse with the utility assigned to a sure dividend $x$, i.e.,
\begin{equation}
\label{eqUtility}
U_j(x) = -e^{-\lambda_j x}\quad (\lambda_j >0),
\end{equation}
that is characterised by a constant absolute risk aversion level $\lambda_j$. We note that the utility function $U:(0,\infty)\rightarrow\mathbb{R}$ in Eq. \eqref{eqUtility} is $C^2$, and satisfies $U'>0$, $U''<0$ as well as the Inada conditions \cite{Inada1963}. Under $U$, the certainty equivalent of $\E[X]$ for agent $j$ is given by
\begin{equation}
\label{eqCertaintyE}
x_c^j = -\frac{\ln\lr-\E\ls U_j(X)|\bar{\sigma}(\xi_t^j)\rs\rr}{\lambda_j}
\end{equation}

with $x_c^j<\E[X|\bar{\sigma}(\xi_t^2)]$ following from strict concavity. Assuming again $X$ is normal with $\N(0,1)$, the equilibrium strategy in a market where agents maximise their expected utility from terminal wealth is now associated to the objective function which is analogous to Eq. \eqref{eqMaxRiskNeutral} and \eqref{eqMaxRiskAdjusted} and given by
\begin{eqnarray}
\label{eqMaxRiskNeutral2}
&\arg\max_{\lr q_t^j \rr} &\displaystyle\sum_{t=0}^T\E \ls U_j(\Pi_t^{j}) \bigr\rvert \bar{\sigma}(\xi_t^j) \rs 
\end{eqnarray}
where each signal-based price $S_t^j$ is worked out, this time, according to certainty equivalence relation in Eq. \eqref{eqCertaintyE} as follows (assuming $s_t=0$ for simplicity):
\begin{eqnarray}
\label{eqCertaintyEPrice}
U_j(S_t^j)=\E\ls U_j(\pm X)\bigr\rvert \xi_t^j\rs&=&-\frac{\displaystyle\int_\mathbb{X} e^{-\lambda_j \pm x} e^{-\frac{x^2}{2}} e^{\kappa_t \lr\sigma_j \xi_{t}^j x - \frac{1}{2} \sigma_j^2 x^2 t\rr } \text{d}x}{\displaystyle\int_\mathbb{X} e^{-\frac{x^2}{2}} e^{\kappa_t \lr\sigma_j \xi_{t}^2 x - \frac{1}{2} \sigma_j^2 x^2 r\rr } \text{d}x}\nn\\
&=&-\frac{\displaystyle\int_\mathbb{X} e^{-\frac{x^2}{2}} e^{\lr\sigma_j \kappa_t \xi_{t}^j \mp \lambda_j \rr x - \lr\frac{1}{2} \sigma_j^2 \kappa_t t\rr  x^2} \text{d}x}{\displaystyle\int_\mathbb{X} e^{-\frac{x^2}{2}} e^{\lr\sigma_j \kappa_t \xi_{t}^j \rr x - \lr\frac{1}{2} \sigma_j^2 \kappa_t t\rr  x^2} \text{d}x}\nn\\
&=&-e^{\frac{(\sigma_j \kappa_t \xi_{t}^j \mp \lambda_j)^2}{2(\sigma_j^2 \kappa_t t +1)}}e^{-\frac{(\sigma_j \kappa_t \xi_{t}^j)^2}{2(\sigma_j^2 \kappa_t t +1)}}\nn\\
&=&-e^{\frac{1}{2}\frac{\lambda_j^2}{\sigma_j^2 \kappa_t t +1} \mp \frac{\sigma_j \kappa_t \xi_{t}^j \lambda_j}{\sigma_j^2 \kappa_t t +1}} .
\end{eqnarray}
This implies, for a bid quote,
\begin{equation}
\label{eqPriceUtility}
S_t^j = U_j^{-1} (\E\ls U_j(X)\bigr\rvert \xi_t^j\rs) = \frac{\sigma_j \kappa_t \xi_{t}^j}{\sigma_j^2 \kappa_t t +1} - \frac{1}{2}\frac{\lambda_j}{\sigma_j^2 \kappa_t t +1}
\end{equation}
and, similarly, for an ask quote,
\begin{equation}
\label{eqPriceUtility2}
S_t^j = -U_j^{-1} (\E\ls U_j(-X)\bigr\rvert \xi_t^j\rs) = \frac{\sigma_j \kappa_t \xi_{t}^j}{\sigma_j^2 \kappa_t t +1} + \frac{1}{2}\frac{\lambda_j}{\sigma_j^2 \kappa_t t +1}.
\end{equation}
The second term in Eq. \eqref{eqPriceUtility} and \eqref{eqPriceUtility2} can be considered as the ``information-adjusted risk premium'' and appears naturally as the bid/ask spread which is inversely proportional to $\sigma_j$ and $t$. Thus, given $\lambda_j$ and $t$, the more an agent is informationally more (less) susceptible, the lower (higher) a risk premium he will have.

The central planner, on his side, will set the price transaction price to the one which equalises their individual signal-based expected utility from the transaction, i.e.,
\begin{equation}
\label{eqPricingRule}
\E\ls U\lr q_t^1 (X-S_t^\ast)\rr | \xi_t^1 \rs = \E\ls U\lr q_t^{2} (X-S_t^\ast)\rr | \xi_t^{2} \rs .
\end{equation}

Assuming again $|q_t^j|\in\{0,1\}$, market is a high-type and agent $2$ buys, the pricing rule in Eq. \eqref{eqPricingRule} can be arranged further as follows:
\begin{eqnarray}
\label{eqClearingRiskAverse0}
-\frac{\displaystyle\int_\mathbb{X} e^{-\lambda_1 (S_t^\ast-x)} e^{-\frac{x^2}{2}} e^{\kappa_t \lr\sigma_1 \xi_{t}^1 x - \frac{1}{2} \sigma_1^2 x^2 t\rr } \text{d}x}{\displaystyle\int_\mathbb{X} e^{-\frac{x^2}{2}} e^{\kappa_t \lr\sigma_1 \xi_{t}^1 x - \frac{1}{2} \sigma_1^2 x^2 t\rr } \text{d}x} &=& -\frac{\displaystyle\int_\mathbb{X} e^{-\lambda_2 (x-S_t^\ast)} e^{-\frac{x^2}{2}} e^{\kappa_t \lr\sigma_2 \xi_{t}^2 x - \frac{1}{2} \sigma_2^2 x^2 t\rr } \text{d}x}{\displaystyle\int_\mathbb{X} e^{-\frac{x^2}{2}} e^{\kappa_t \lr\sigma_2 \xi_{t}^2 x - \frac{1}{2} \sigma_2^2 x^2 r\rr } \text{d}x}\nn\\
-\frac{e^{-\lambda_1 S_t^\ast}\displaystyle\int_\mathbb{X} e^{-\frac{x^2}{2}} e^{\lr\sigma_1 \kappa_t \xi_{t}^1 + \lambda_1 \rr x - \lr\frac{1}{2} \sigma_1^2 \kappa_t t\rr  x^2} \text{d}x}{\displaystyle\int_\mathbb{X} e^{-\frac{x^2}{2}} e^{\lr\sigma_1 \kappa_t \xi_{t}^1\rr x - \lr\frac{1}{2} \sigma_1j^2 \kappa_t t\rr  x^2} \text{d}x} &=& -\frac{e^{\lambda_2 S_t^\ast}\displaystyle\int_\mathbb{X} e^{-\frac{x^2}{2}} e^{\lr\sigma_2 \kappa_t \xi_{t}^2 - \lambda_2 \rr x - \lr\frac{1}{2} \sigma_2^2 \kappa_t t\rr  x^2} \text{d}x}{\displaystyle\int_\mathbb{X} e^{-\frac{x^2}{2}} e^{\lr\sigma_2 \kappa_t \xi_{t}^2\rr x - \lr\frac{1}{2} \sigma_2^2 \kappa_t t\rr  x^2} \text{d}x}\nn\\
-e^{-\lambda_1 S_t^\ast} e^{\frac{(\sigma_1 \kappa_t \xi_{t}^1 + \lambda_1)^2}{2(\sigma_1^2 \kappa_t t +1)}}e^{-\frac{(\sigma_1 \kappa_t \xi_{t}^1)^2}{2(\sigma_1^2 \kappa_t t +1)}}&=&-e^{\lambda_2 S_t^\ast} e^{\frac{(\sigma_2 \kappa_t \xi_{t}^2 - \lambda_2)^2}{2(\sigma_2^2 \kappa_t t +1)}}e^{-\frac{(\sigma_2 \kappa_t \xi_{t}^2)^2}{2(\sigma_2^2 \kappa_t t +1)}}\nn\\
-e^{-\lambda_1 S_t^\ast + \frac{1}{2}\frac{\lambda_1^2}{\sigma_1^2 \kappa_t t +1} + \frac{\sigma_1 \kappa_t \xi_{t}^1 \lambda_1}{\sigma_1^2 \kappa_t t +1}}&=&-e^{\lambda_2 S_t^\ast + \frac{1}{2}\frac{\lambda_2^2}{\sigma_2^2 \kappa_t t +1} - \frac{\sigma_2 \kappa_t \xi_{t}^2 \lambda_2}{\sigma_2^2 \kappa_t t +1}},
\end{eqnarray}
which directly implies, also by Eq. \eqref{eqPriceUtility}, that
\begin{eqnarray}
\label{eqClearingRiskAverse}
S_t^\ast &=& \frac{\lambda_1\lr\frac{\sigma_1 \kappa_t \xi_{t}^1}{\sigma_1^2 \kappa_t t +1}+\frac{1}{2}\frac{\lambda_1}{\sigma_1^2 \kappa_t t +1}\rr + \lambda_2\lr\frac{\sigma_2 \kappa_t \xi_{t}^2}{\sigma_2^2 \kappa_t t +1} - \frac{1}{2}\frac{\lambda_2}{\sigma_2^2 \kappa_t t +1}\rr}{\lambda_1+\lambda_2}\nn\\
&=& \frac{\lambda_1}{\lambda_1+\lambda_2} S_t^1 + \frac{\lambda_2}{\lambda_1+\lambda_2} S_t^2.
\end{eqnarray}

Thus, the information-based market price is the weighted average of risk-adjusted signal-based prices with respect to the risk aversion levels $\lambda_j$. Accordingly, $a_t^s$ and $b_t^s$, and all relevant quantities such as $P_t^j (\xi_c^j | x+,s)$ will need to be updated to re-explore an optimal strategy.

\section{Conclusion}

We've introduced a network of a pair of agents with heterogeneous signal sources (or signal processing skills). The gap between the agents' stepwise average P\&Ls is shown to be directly linked to whether agents mutually learn during trade. The case where agents are `attentive', i.e., do learn from each other, as compared to the case where they are `omitters,' is associated with a significant convergence of P\&Ls. It is also apparent, from the analysis on the impact of mutual learning on how quickly the true value of the fundamental is discovered, that mutual learning works in favour of the agent with an inferior signal, and it turns out to be a win-win situation when they have similar informational skills. 

Furthermore, the existence of a `common knowledge of gains from trade' (in the sense of \cite{BondEraslan2010}) is found to be essential to an equilibrium, and to avoid market shutdowns. We derive explicit formula for the ex ante P\&L of the agent before an auction takes place through quantifying the ex-ante directional quality of his trading signal and his expected P\&L conditional on whether the signal points is correct or erroneous. As expected, perception of a greater informational superiority means a greater likelihood for the agent that his trading signal is directionally correct and greater expected profits. And, this likelihood is stronger in the case of an a priori greater uncertainty (i.e., greater dispersion of the payoff's prior density), and also whenever the agent chooses to refrain from trade. In equilibrium, we find that the optimal strategy is to exploit extra information as soon as it arrives, as the cost of foregoing an expected profit is higher than the cost of sharing the new information. This result not only agrees with most of the informed-trading literature, but also establishes, in a soft sense, the rationale behind high-frequency trading.

Finally, when agents are risk-averse, we are able to recover the risk-premium as inversely proportional to the information flow rate, which allows us to introduce the concept of `information-adjusted risk premium'. The equilibrium price in this case is a nice convex combination of individual signal-based prices with respect to individual levels of risk aversion.

As an outlook, we would like to extend the present analysis to cases where there are multiple agents, the information flows are mimicked by a more general class of L\'{e}vy processes, possibly including jumps, and the amount of information shared is a function of the amount traded.


\bibliographystyle{plain}
\bibliography{NSA_arxiv}
\end{document}